\documentclass[lettersize,journal]{IEEEtran}
\usepackage{cite}
\usepackage{amsmath, amsthm, amsfonts, amssymb}
\newtheorem{theorem}{Theorem}
\newtheorem{proposition}[theorem]{Proposition}
\newtheorem{definition}[theorem]{Definition}
\newtheorem{lemma}[theorem]{Lemma}
\theoremstyle{definition}
\newtheorem{remark}{Remark}
\usepackage[hang,flushmargin]{footmisc}
\usepackage[linesnumbered,ruled]{algorithm2e}
\usepackage{bbm}
\usepackage{bm}
\usepackage{array}
\usepackage{textcomp}
\usepackage{stfloats}
\usepackage{url}
\usepackage{verbatim}
\usepackage{graphicx}
\usepackage{setspace}
\usepackage{subfigure}%
\usepackage{graphicx}
\usepackage{float}%
\usepackage{booktabs}
\usepackage{multirow}
\usepackage{array}
\usepackage{enumerate}
\usepackage{bbm}
\usepackage{textcomp}
\usepackage{xcolor}
\usepackage{stfloats}
\usepackage{verbatim}
\usepackage{tikz}
\usepackage{soul}
\usepackage{pifont}
\usepackage{amsmath}   
\usepackage{amsthm}
\usepackage{tabularx}
\usepackage{enumitem}
\usepackage{booktabs}   
\usepackage{siunitx}    
\usepackage{amsmath}   
\usepackage{makecell}
\usepackage{multirow}
\usepackage[normalem]{ulem}

\usepackage{xfrac}
\usepackage{nicefrac}

\begin{document}

\title{Age of Information for Constrained Scheduling with Imperfect Feedback}
\author{Yuqing Zhu, Yuan-Hsun Lo, Yan Lin and Yijin Zhang

\thanks{\quad This work was supported in part by the National Natural Science Foundation of China under Grant 62071236. 
(\emph{Corresponding author: Yijin Zhang}.)}
\thanks{\quad Y. Zhu, Y. Lin, and Y. Zhang are with the School of Electronic and Optical Engineering, Nanjing University of Science and Technology, Nanjing 210094, China (e-mail: \{yuqing.zhu; yanlin\}@njust.edu.cn; yijin.zhang@gmail.com).}
\thanks{\quad Y.-H. Lo is with the Department of Applied Mathematics, National Pingtung University, Pingtung 90003, Taiwan (e-mail: yhlo@mail.nptu.edu.tw).}
}

\maketitle

\begin{abstract}
This paper considers a downlink system where an access point sends the monitored status of multiple sources to multiple users.
By jointly accounting for imperfect feedback and constrained transmission rate, which are key limited factors in practical systems, we aim to design scheduling algorithms to optimize the age of information (AoI) over the infinite time horizon.
For zero feedback under the generate-at-will traffic, we derive a closed-form lower bound of achievable AoI, which, to the best of our knowledge, reflects the impact of zero feedback for the first time, and propose a policy that achieves this bound in many cases by jointly applying rate splitting and modular arithmetic.
For zero feedback under the Bernoulli traffic, we develop a drift-plus-penalty (DPP) policy with a threshold structure based on the theory of Lyapunov optimization and provide a closed-form performance guarantee. 
Furthermore, we extend the design of this DPP policy to support general imperfect feedback without increasing the online computational complexity.
Numerical results verify our theoretical analysis and the AoI advantage of the proposed policies over state-of-the-art policies. 
\end{abstract}


\section{Introduction}\label{Introduction}

Age of Information (AoI)~\cite{2011Minimizing} has been commonly recognized as a standard performance metric to evaluate the information freshness of time-sensitive services in Internet of Things (IoT).
By definition, it measures the time elapsed since the generation time of the update packet most recently delivered to the destination, which has been shown to be fundamentally different from traditional metrics, such as throughput or delay.
Thus, much research attention has been devoted to optimizing AoI under various system settings~\cite{Yates2021}.
This paper focuses on packet scheduling for downlink systems where an access point (AP) sends the monitored status of multiple sources to multiple users over a shared channel.

Feedback availability has a significant impact on achievable AoI. 
Most studies on downlink scheduling assumed instant error-free feedback, which allows the AP to perform more reasonable scheduling actions with the aid of immediate transmission outcomes~\cite{Kadota2018ACM,Kadota2021TMC}. 
However, implementing such perfect feedback may be infeasible for multi-cast, long-distance, small-packet 
or low-cost communications~\cite{Nguyen2022,LoRAWC}, and 
unreliable feedback channels~\cite{Rezasoltani2022TWC,PyttelICC2024}, so that the AP has to make scheduling decisions with outdated or even without transmission outcomes. 
So, it is important to minimize the AoI loss due to the lack of perfect feedback through efficiently utilizing locally available knowledge.
Under this objective,~\cite{Bacinoglu2017ISIT,Feng2021tcom,Munari2025TCOM} considered traffic and action histories and a priori statistical information for zero-feedback scenarios, while~\cite{Lu2023ICCC,Ji2024TCOM,ZhaoINFOCOM2025, Rezasoltani2022TWC,PyttelICC2024} additionally considered feedback history for delayed- or erroneous-feedback scenarios.
Obviously, the lack of perfect feedback
would lead to higher modeling complexity and more difficulty in efficiently utilizing locally available knowledge with low complexity, which involves the application of Markov Decision Process (MDP)~\cite{Munari2025TCOM}, Partially-Observable MDP (POMDP)~\cite{PyttelICC2024,Lu2023ICCC,Ji2024TCOM}, Lyapunov optimization~\cite{ZhaoINFOCOM2025}, and offline sequence design~\cite{Li2021INFOCOM,Liyanaarachchi2025TIT}.

Another important factor influencing achievable AoI is the allowable maximum long-term average transmission rate.
Naturally, it is desirable to let it take the value $1$.
However, due to hardware and energy limitations, it is necessary to consider transmission rate constraints, which obviously complicates the policy design.
For a single source in finite-horizon scenarios, graphical analysis and solving linear equations are used in~\cite{Munari2025TCOM} to find a near-optimal zero-feedback policy under the generate-at-will (GAW) traffic; 
the theory of MDP is used in~\cite{Munari2025TCOM} to find an optimal perfect-feedback policy under the Bernoulli traffic, which also provides clues to design a zero-feedback policy.
However, these methods \textit{cannot} be used in infinite-horizon scenarios without performance loss,
since an infinite-horizon transmission rate constraint \textit{cannot} be exactly converted into a finite number of transmission instants or embedded into an MDP state that provides all necessary information for decision.
To deal with this issue, the most common method is to use the Constrained MDP (CMDP) framework~\cite{Elif2019TWC,Ceran2021JSAC}, 
but solving it through the direct primary formulation or the Lagrangian dual approach usually requires finite states, rigid applying conditions, and high computational complexity.
So, Lagrange-cost based greedy approaches\cite{Ji2024TCOM}, Lyapunov optimization based approaches~\cite{Vilni2022}, token MDP based approaches~\cite{1DelfaniTCOM2025}, and reinforcement learning based approaches~\cite{Ceran2021JSAC} have been proposed to address more infinite-horizon models, although lacking optimality guarantees.
The work\cite{Elif2019TWC,Ceran2021JSAC,Vilni2022,1DelfaniTCOM2025} all assumed perfect feedback, but \textit{cannot} be applied without perfect feedback due to the significant difference in locally available knowledge, while the work~\cite{Ji2024TCOM} assumed delayed feedback that can include zero feedback as a particular case and sequence-based approaches under transmission rate constraints~\cite{Li2021INFOCOM,Liyanaarachchi2025TIT} are also applicable for zero feedback.
Note that~\cite{PyttelICC2024} considered erroneous feedback under a real-time constraint due to energy harvesting, which can be formulated by POMDPs rather than CMDPs, but this approach \textit{cannot} be applied under a long-term average constraint.

Based on the aforementioned discussions, it is expected that jointly considering imperfect feedback and constrained transmission rate would make AoI optimization over the infinite horizon more technically challenging, which involves the following fundamental questions.

    (i) How to establish a lower bound of the AoI performance that reflects the negative impact of these two factors?
    
    (ii) How to efficiently utilize locally available knowledge with low complexity for decision making?
    
    (iii) How to use suitable approaches with low complexity to design optimal or near-optimal policies?

As summarized in Table~\ref{Comparison of our work and related work}, these questions have been only partially investigated.
Thus, we make an attempt to fill the gap.
The contributions of this paper are summarized as follows.

    (i) \textit{Zero Feedback:} Under the GAW traffic, we derive a closed-form lower bound of achievable AoI for infinite-horizon scenarios.
     To the best of our knowledge, this is the first lower bound for an arbitrary number of sources that reflects the impact of zero feedback, which involve a new proving approach different from~\cite{Kadota2018ACM,Elif2019TWC,Kadota2021TMC}.
    Furthermore, we propose a policy that achieves this bound in many cases, by extending an offline sequence-based scheme in~\cite{Li2021INFOCOM,Liyanaarachchi2025TIT}.
    Our design is based on an idea of jointly applying rate splitting and modular arithmetic.
    We also provide a lower bound for finite-horizon scenarios as a by-product, which is unknown in the literature.
   
    (ii) \textit{Zero Feedback:} Under the Bernoulli traffic, by defining a hybrid Lyapunov function that characterizes the transmission rate usage quadratically and the AoI evolution linearly, we develop a drift-plus-penalty (DPP) policy based on the theory of Lyapunov optimization.
    The DPP policy enjoys a simple threshold structure, which uses the conditional expected age-based weight~\cite{Kadota2021TMC} (also called age gain in~\cite{Chen2022TIT}) together with the transmission rate usage to measure the transmission preference.
    We further provide a closed-form performance guarantee with a connection to an optimal stationary randomized policy.
    
    (iii) \textit{Imperfect Feedback:} We extend the design of the DPP policy to account for general imperfect feedback that incorporates feedback delays and feedback errors. 
    We use the Bayesian rule to derive the conditional expected AoI based on given feedback information and feedback mechanism, without increasing the online computational complexity compared with that for zero feedback.
    Our derivation generalizes that for zero~\cite{Munari2022GLOBECOM}, delayed~\cite{Ji2024TCOM}, and erroneous feedback~\cite{PyttelICC2024}.

Throughout this paper, we compare our work with the most related work~\cite{Kadota2018ACM,Kadota2021TMC,Munari2025TCOM,Feng2021tcom,Li2021INFOCOM,Ji2024TCOM,Liyanaarachchi2025TIT,ZhaoINFOCOM2025,Ceran2021JSAC,PyttelICC2024} in optimization tools or in theoretical results to illustrate the impact of feedback, constrained transmission rate, and other factors.
Our work also includes many results in~\cite{Kadota2021TMC,Li2021INFOCOM,Ji2024TCOM,Munari2025TCOM,PyttelICC2024} as particular cases.
Please see Remarks~\ref{remarkfinitehorizon}--\ref{reducetomax-weight} for more details.
Numerical results demonstrate that the proposed policies outperform the state-of-the-art policies~\cite{Ji2024TCOM,Liyanaarachchi2025TIT}.

The rest of this paper is organized as follows. 
We start in Section~\ref{sec:SystemModel} by setting up the system model.
Sections~\ref{Zero-Feedback Scenarios: GAW Traffic} and~\ref{Zero Feedback: Bernoulli Traffic} investigate zero-feedback scenarios under the GAW and Bernoulli traffic, respectively. 
Section~\ref{Extension to Imperfect-Feedback Scenarios} expands to consider general imperfect feedback.
Numerical results are provided in Section~\ref{Numerical Results}. 
We draw the conclusions in Section~\ref{Conclusions}.

\begin{table}[ht] 
\centering  
\footnotesize
\caption{Comparison of our work and related work. $*$ indicates that only point-to-point scenarios are considered, while $\circ$ indicates that the work can be easily extended to support.}
\label{Comparison of our work and related work}
\begin{tabular}{@{\hspace{1pt}}c@{\hspace{2.5pt}}c@{\hspace{7pt}}c@{\hspace{7pt}}c@{\hspace{7pt}}c@{\hspace{5pt}}c}
\toprule 
\multicolumn{2}{c}{}
& {\bf{\makecell{\scriptsize{Time} \\ \scriptsize{horizon}}}}    &{\bf{\makecell{\scriptsize{Traffic}}}}   & {\bf{\makecell{\scriptsize{Feedback}}}} & {\bf{\makecell{\scriptsize{Long-term} \\ \scriptsize{average} \\ \scriptsize{constraint?}}}} \\
\midrule
\multirow{3}{*}[-1.5ex] {\bf{\makecell[c]{\scriptsize{Lower}\\ \scriptsize{bound}}}} 
&{ \cite{Ceran2021JSAC} } &infinite &GAW &perfect &\checkmark \\ \cmidrule(lr){2-6}
&{ \cite{Kadota2021TMC} } &infinite &Bernoulli &perfect  &$\checkmark^\circ$ \\ \cmidrule(lr){2-6}
&{ \makecell[c]{Ours}}  &infinite &GAW &zero &$\checkmark$ \\ \midrule
\multirow{3}{*}[-15ex] {\bf{\makecell[c]{\scriptsize{Policy} \\ \scriptsize{design}}}} 
&{\makecell{\cite{Bacinoglu2017ISIT}*}}  &finite &GAW &zero &\ding{55} \\ \cmidrule(lr){2-6}
&{\makecell{\cite{PyttelICC2024}*}}  &finite &GAW &erroneous &\ding{55} \\ \cmidrule(lr){2-6}
&{\makecell{\cite{Munari2025TCOM}*}}  &finite &Bernoulli &perfect/zero &$\checkmark$ \\ \cmidrule(lr){2-6}
&{\makecell{\cite{Elif2019TWC}*, \cite{Ceran2021JSAC}}}  &infinite &GAW &perfect &\checkmark \\ \cmidrule(lr){2-6}
&{\makecell{\cite{1DelfaniTCOM2025}*}}  &infinite &Bernoulli &perfect &\checkmark \\ \cmidrule(lr){2-6}
&{\makecell{~\cite{Liyanaarachchi2025TIT}}} &infinite &GAW &zero &$\checkmark^\circ$  \\ \cmidrule(lr){2-6}
&{\makecell{\cite{Rezasoltani2022TWC}*}}  &infinite &Bernoulli &erroneous &\ding{55} \\ \cmidrule(lr){2-6}
&{\makecell{\cite{Lu2023ICCC},~\cite{ZhaoINFOCOM2025}}}  &infinite &Bernoulli &delayed &\ding{55} \\ \cmidrule(lr){2-6}
&{\makecell{\cite{Ji2024TCOM}}}  &infinite &Bernoulli &delayed &$\checkmark$ \\ \cmidrule(lr){2-6}
&{\makecell[c]{Ours}} &infinite &Bernoulli &\makecell{zero/delayed/erroneous} &$\checkmark$ \\ 
\bottomrule
\end{tabular}
\end{table}

\section{System Model}  \label{sec:SystemModel}

\subsection{Network description}  \label{sec: Network description}
As shown in Fig.~\ref{systemmodel}, consider a downlink system~\cite{Kadota2021TMC,Ji2024TCOM,ZhaoINFOCOM2025,Liyanaarachchi2025TIT} where $N$ users indexed by $\mathcal{N}\triangleq\{1,2,\ldots, N\}$ want to receive state updates from corresponding $N$ sources indexed by $\mathcal{N}$ via an AP.
The channel time is divided into equal-length slots, indexed by $t\in\mathbb{N}$.
At the beginning of each slot $t$, each source $n$ independently generates a single-slot packet with a fixed probability $\lambda_n \in (0,1]$.
Let $d_{n,t} \in \{0,1\}$ denote the packet generation indicator for source $n$ in slot $t$, where $d_{n,t} = 1$ if source $n$ generates a new packet at the beginning of slot $t$, and $d_{n,t} = 0$ otherwise.
To maintain information freshness, any newly generated packet at each source immediately replaces the older one.

\begin{figure}[!ht]
	\centering
	\includegraphics[height = 1.05in]{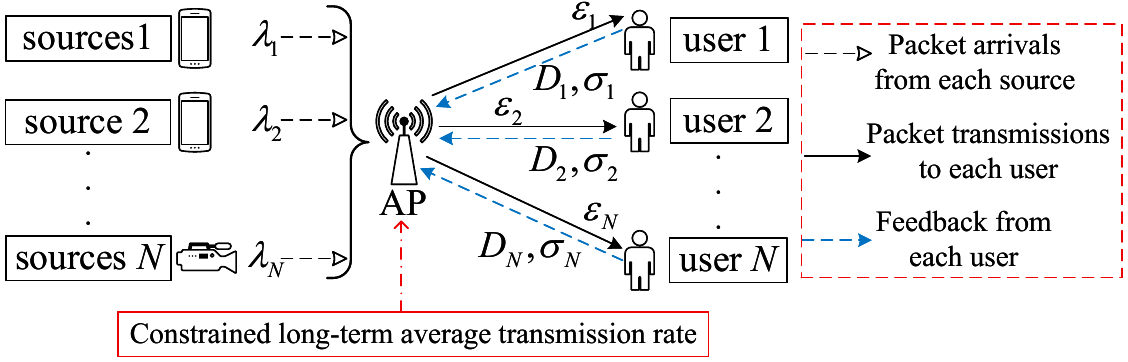}
	\caption{Considered downlink scenario.}
	\label{systemmodel}
\end{figure}

Let $a_{n,t}\in\{0,1\}$ denote the scheduling action for source $n$ in slot $t$, where $a_{n,t} = 1$ if source $n$ is scheduled at the beginning of slot $t$, and $a_{n,t} = 0$ otherwise.
For a shared channel, we require $\sum_{n \in \mathcal{N}} a_{n,t}\leq 1$ for each $t$.
Let $u_{n,t}\in\{0,1\}$ denote the transmission outcome of source $n$ in slot $t$, where $u_{n,t} = 1$ if a packet is successfully received by user $n$ at slot $t$, and $u_{n,t} = 0$ otherwise.
We consider an unreliable downlink channel with a fixed error probability $\varepsilon_n \in [0,1)$ between the AP and user~$n$, so that $u_{n,t} = 0$ with probability $\varepsilon_n$ when $a_{n,t} = 1$ and $u_{n,t} = 0$ certainly when $a_{n,t} = 0$.

Assume that the feedback channel for each user $n$ has a fixed error probability $\sigma_n\in [0,1]$ and a fixed delay $D_n\in\mathbb{N}$.
We further consider the following two feedback mechanisms.

    (i) \textit{ACKs mechanism:} 
    Only after each successful transmission, an acknowledgment (ACK) is sent back to the AP to notify the transmission outcome.
    Let $v_{n,t} \in \{0,1\}$ denote the feedback indicator from user $n$ at the end of slot $t-1$, 
where $v_{n,t} = 1$ if the AP received the delayed ACK from user $n$ that implies $u_{n, t - D_n - 1}=1$, but $v_{n,t} = 0$ otherwise.
We have $v_{n,t} = 0$ with probability $1 - (1 - \varepsilon_n)(1 - \sigma_n)$ when $a_{n,t - D_n - 1} = 1$, and $v_{n,t} = 0$ certainly when $a_{n,t - D_n - 1} = 0$.

    (ii) \textit{ACKs/NACKs mechanism:} 
    Whenever a transmission is performed, an ACK is sent back if the transmission is successful, but a negative ACK (NACK) is sent back otherwise. 
    Let $v_{n,t} \in \{0,1,-1\}$ denote the feedback indicator from user $n$ at the end of slot $t-1$, 
where $v_{n,t} = 1$ if the AP received the delayed ACK from user $n$ that implies $u_{n,t - D_n - 1}=1$, $v_{n,t} = -1$ if the AP received the delayed NACK from user $n$ that implies $a_{n,t - D_n - 1} = 1, u_{n,t - D_n - 1} = 0$, and $v_{n,t} = 0$ otherwise.
We have $v_{n,t} = 0$ with probability $\sigma_n$ when $a_{n,t - D_n - 1}=1$, and $v_{n,t} = 0$ certainly when $a_{n,t - D_n - 1} = 0$.

Note that $\sigma_n=1$ or $D_n=\infty$ corresponds to zero feedback for source $n$, which leads to $v_{n,t}=0$ for each $t\in\mathbb{N}$.

Denote the local age of source $n$ at the beginning of slot $t$ by $w_{n,t}$, which measures the number of slots elapsed since the generation moment of its freshest packet.
The local age of source $n$ is reset to zero if source $n$ generates a new packet at the beginning of slot $t$, and increases by one otherwise.
Then, the evolution of $w_{n,t}$ with $w_{n,0} = 0$ is given by
\begin{equation}
	\label{eq:Evolution_w}
	{w_{n,t+1}} =
	\begin{cases}
		{0,}            &\text{if }d_{n,t+1} = 1, \\ 
		{w_{n,t}+1,}    &\text{if }d_{n,t+1} = 0. 
	\end{cases}
\end{equation}

Denote the AoI associated with user $n$ at the beginning of slot $t$ by $h_{n,t}$, which measures the number of slots elapsed since the generation moment of its most recently received packet.
If the freshest packet of source $n$ is successfully transmitted in slot $t$, the AoI of user $n$ will be set to its local age in the previous slot plus one; otherwise, it will increase by one.
The evolution of $h_{n,t}$ with $h_{n,0} = 1$ is given by
\begin{equation}
	\label{eq:Evolution_AoI}
	h_{n,t+1} =
	\begin{cases}
		{w_{n,t}+1},       &\text{if }u_{n,t} = 1,\\ 
		{h_{n,t} + 1},     &\text{if }u_{n,t} = 0.
	\end{cases}
\end{equation}

\subsection{Problem Formulation}\label{problem formulation}

The scheduling policies considered in this paper are non-anticipative, i.e., policies that do not use future knowledge in making scheduling decisions.
Let $\bm{\Pi}$ denote the set of all non-anticipative policies.
We define the infinite-horizon Expected Weighted Sum AoI (EWSAoI) induced by $\bm{\pi}\in \bm{\Pi}$ as follows:
\begin{align}\label{eq: Delta}
J^{\bm{\pi}}  \triangleq 
\limsup_{T\to\infty}\frac{1}{T} \sum_{t=0}^{T-1}\sum_{n \in \mathcal{N}} \alpha_n \mathbb{E}^{\bm{\pi}}&\big[ h_{n,t} \mid h_{n,0} = 1, \notag \\ & w_{n,0} = 0, \forall{n \in \mathcal{N} } \big],
\end{align}
where the expectation $\mathbb{E}^{\bm{\pi}}[\cdot]$ is taken over the randomness of the system and the transmission actions under a policy $\bm{\pi} \in \bm{\Pi}$, the normalized source weight $\alpha_n>0$ represents the priority of source $n$ with $\sum_{n \in \mathcal{N}} \alpha_n=1$.
The infinite-horizon transmission rate induced by $\bm{\bm{\pi}} \in \bm{\Pi}$ is defined as
\begin{align}  
q^{\bm{\pi}}  \triangleq \limsup_{T\to\infty}\frac{1}{T}\sum_{t=0}^{T-1}\sum_{n \in \mathcal{N}}\mathbb{E}^{\bm{\pi}}&\big[a_{n,t} \mid h_{n,0} = 1,\notag  \\ & w_{n,0} = 0,\forall n \in \mathcal{N}  \big ].
\label{eq:Gamma}
\end{align}
Using~\eqref{eq: Delta} and~\eqref{eq:Gamma}, the optimization problem is formulated as
\begin{equation}
    \begin{aligned}
        \min_{\bm{\pi} \in \bm{\Pi}}  {J}^{\bm{\pi}} , \quad
        \text{ s.t. }   q^{\bm{\pi}}  \leq \rho,
    \label{problem1}
\end{aligned}
\end{equation}
where $\rho$ denotes the allowable maximum transmission rate.

\section{Zero-Feedback Scenarios: GAW Traffic}
\label{Zero-Feedback Scenarios: GAW Traffic}

This section concentrates on zero-feedback scenarios under the GAW traffic.
We derive a lower bound of the infinite-horizon EWSAoI and propose a policy that achieves this bound in many cases.

\subsection{Lower Bound}
To the best of our knowledge, this is the first lower bound that reflects the impact of zero feedback.
Our approach is different from those in~\cite{Kadota2018ACM,Kadota2021TMC} that rely on a given single sample path and applying Jensen's inequality to obtain the minimum of the sample variance of inter-successful-delivery time, which may be infeasible under zero feedback.

To facilitate our analysis, we relax the constraint in~\eqref{problem1} and decouple the model to $N$ point-to-point systems.
Let $\bm{\Pi^\text{single}}$ denote the set of all non-anticipative policies for a generic point-to-point system.
Specifically, we assume that in each decoupled system $n$, source $n$ has a transmission rate constraint $\rho_n \in [0,\rho]$ for each $n\in\mathcal{N}$, with $\sum_{n\in\mathcal{N}} \rho_n = \rho$.
For any $\pi^\text{single} \in \bm{\Pi^\text{single}}$, we define 
\[
J_n^{\bm{\pi}^\text{single}} \triangleq \limsup_{T\to\infty}\frac{1}{T} \sum_{t=0}^{T-1} \mathbb{E}^{\bm{\pi}^\text{single}}\left[ h_{n,t} \mid h_{n,0} = 1, w_{n,0} = 0 \right],
\]
\[
q_n^{\bm{\pi}^\text{single}} \triangleq \limsup_{T\to\infty}\frac{1}{T}\sum_{t=0}^{T-1}\mathbb{E}^{\bm{\pi}^\text{single}}\left[a_{n,t} \mid h_{n,0}\! =\! 1, w_{n,0}\! =\! 0 \right].
\]
Thus, we reduce problem~\eqref{problem1} to $N$ single-source problems, each of which can be given by
\begin{align}  \label{problem_n}
   \min_{\bm{\pi^\text{single}} \in \bm{\Pi^\text{single}}}  J_n^{\bm{\pi}^\text{single}}, \quad \text{ s.t. }  q_n^{\bm{\pi}^\text{single}} \leq \rho_n.
\end{align}
Denote by $J_{n}^{\text{single}*}$ the value of~\eqref{problem_n}. 
Then, we have
\begin{align}  \label{Jpi}
    J^{\bm{\pi}} \geq \sum_{n\in\mathcal{N}}\alpha_n J_{n}^{\text{single}*},
\end{align}
for any $\bm{\pi}\in \bm{\Pi}$ with $q^{\bm{\pi}} \leq \rho$.

In the following, we focus on a generic decoupled point-to-point system $n$. 
Over the finite time horizon $[0, T-1]$, consider those sample paths associated with a policy $\bm{\pi}^\text{single}\in \bm{\Pi}^\text{single}$, where total $U_n(T)$ transmissions are conducted at slots $l_{n,1}, l_{n,2}, \ldots, l_{n,U_n(T)}$. 
As shown in Fig.~\ref{fig:exampleofh}, we let $X_{n,k} \triangleq l_{n,k} + 1 - (l_{n,k-1}+1)$ denote the $k^{\text{th}}$ transmission interval with $l_{n,0} = -1$ and $l_{n,U_n(T)+1} = T - 1$ for $k \in \{1, 2, \ldots,U_n(T)+1\}$.
Define $R_n(T)$ as the cumulative AoI over $[0, T-1]$.
Then, under any policy $\bm{\pi} \in \bm{\Pi}$, following~\cite{Munari2025TCOM}, the conditional expected average AoI over $[0, T-1]$ given $\{l_{n,k}\}_{k=1}^{U_n(T)} $ can be expressed as
\begin{align}
 &  \mathbb{E}^{\bm{\pi}^\text{single}}\left[\frac{R_n(T)}{T}\mid \{l_{n,k}\}_{k=1}^{U_n(T)} \right]= \notag \\
 &   \frac{1}{T} \sum_{k = 1}^{U_n(T)+1} \left[ \frac{X_{n,k}(X_{n,k} + 1)}{2} + \sum_{k' = 1}^{k-1} X_{n,k} X_{n,k'} \varepsilon_n^{k - k'} \right],
\label{roleoffeedback} 
\end{align}
where $ \frac{X_{n,k}(X_{n,k} + 1)}{2}$ accounts for the growth of AoI during the $k^{\text{th}}$ transmission interval of length $X_{n,k}$, 
whereas $X_{n,k} X_{n,k'}$ captures a rectangle of sides $X_{n,k}$ and $X_{n,k'} $ that represents the additional AoI due to a sequence of consecutive failures spanning from the ${k'}^{\text{th}}$ to the $(k-1)^{\text{th}}$ transmission, weighted by the probability $ \varepsilon_n^{k - k'}$.
\begin{figure}
\centering
\includegraphics[width = 3in]{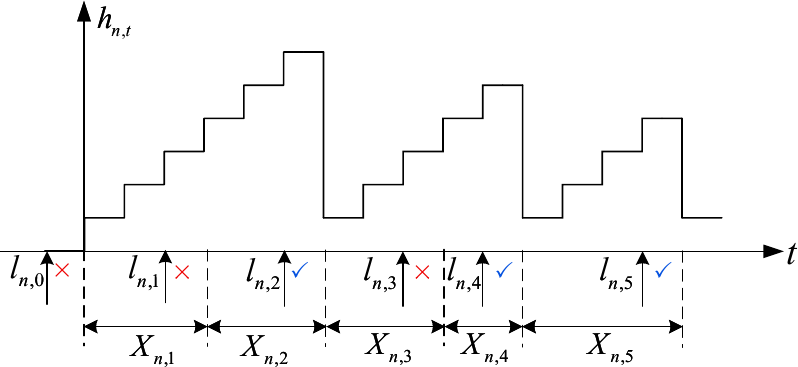}
    \caption{An example of evolution of $h_{n,t}$ for a generic point-to-point system $n$ under the GAW traffic.
    The symbol $\color{red}{\times}$ indicates the transmission failed, whereas $\color{blue}{\checkmark}$ indicates the transmission was successful.}  \label{fig:exampleofh} 
\end{figure}
Rearranging the terms in~\eqref{roleoffeedback}, we obtain the following quadratic form:
\begin{align}
& \mathbb{E}^{\bm{\pi}^\text{single}}\left[\frac{R_n(T)}{T}\mid \{l_{n,k}\}_{k=1}^{U_n(T)} \right]\notag \\
& =  \frac{1}{2T} \sum_{k = 1}^{U_n(T)+1} \left[ X_{n,k}^2 + 2\sum_{k' = 1}^{k-1} X_{n,k} X_{n,k'} \varepsilon_n^{k - k'}  \right] + \frac{1}{2} \notag \\
& = \frac{1}{2T} \sum_{k=1}^{U_n(T)+1} \sum_{k'=1}^{U_n(T)+1} \varepsilon_n^{|k-k'|} X_{n,k}  X_{n,k'} + \frac{1}{2}. \label{eq: ERT}
\end{align}
Finding optimal $\{X_{n,k}\}_{k=1}^{U_n(T)+1}$ that minimizes~\eqref{eq: ERT} is NP-complete~\cite[Chapter 15]{1998NP-complete}. 
Hence, we relax $\{X_{n,k}\}_{k=1}^{U_n(T)+1}$ into continuous variables $\{\hat{X}_{n,k}\}_{k=1}^{U_n(T)+1}$ and formulate the following optimization problem:
\begin{equation}  \label{fxproblem}
    \begin{aligned}
    \min_{ {\hat{\bm{X}}_n} }   \quad &{f_n}({\hat{\bm{X}}_n}) = \frac{1}{2T} {\hat{\bm{X}}}_n^{\text{T}} \bm{A}_n {\hat{\bm{X}}_n} + \frac{1}{2}, \\
    \text{s.t.}\ \quad &\bm{e}^{\text{T}}{\hat{\bm{X}}_n} = T, \\
     & \hat{X}_{n,k} \in [0, T-1], \forall{k} \in \{1,2,\ldots,U_n(T)+1\},
\end{aligned}
\end{equation}
where ${\hat{\bm{X}}_n} \triangleq (\hat{X}_{n,1}, \hat{X}_{n,2}, \dots, \hat{X}_{n,U_n(T)+1})^{\text{T}}$, $\bm{A}_n$ is a $(U_n(T) +1) \times (U_n(T)+1)$-dimensional matrix whose $(i,j)$-th entry is $\varepsilon_n^{|i-j|}$ for each $i,j\in\{1,2,\ldots,U_n(T)+1\}$,
and $\bm{e}$ is a $( U_n(T) + 1 ) \times 1$-dimensional all-ones vector.
Positive definiteness of $\bm{A}_n$ guarantees strict convexity of ${f}_n({\hat{\bm{X}}_n})$, implying that there exists a unique global minimum of $f_n$~\cite[Chapter 7]{1985Matrix}.
In the following, we will derive the minimum of ${f}_n({\hat{\bm{X}}_n})$, denoted by $f_n^*$, which can serve as a lower bound for the original expression with discrete variables~\eqref{eq: ERT}.

\begin{lemma}  \label{minimumfn}

For each $n\in \mathcal{N}$, the minimum of ${f_n}({\hat{\bm{X}}_n})$ is
\begin{align}  \label{f*_finitehorizon}
f_n^* = \frac{T}{2} \frac{1+\varepsilon_n}{2 + (U_n(T) - 1)(1-\varepsilon_n)} + \frac{1}{2}.
\end{align}
\end{lemma}
\begin{proof}
   See Appendix~\ref{prooflemmaminimumf}.
\end{proof}

\begin{remark}  \label{remarkfinitehorizon}
\cite{Munari2025TCOM} obtained the solution of ${\hat{\bm{X}}_n}$ to problem~\eqref{fxproblem} through solving a full-rank system of $U_n(T)$ linear equations in $U_n(T)$ unknowns, with time complexity $\mathcal{O}((U_n(T))^3)$, which leads to no explicit $f_n^*$.
In contrast, we can derive the {closed-form} solution of ${\hat{\bm{X}}_n}$ as shown in Appendix~\ref{prooflemmaminimumf} to problem~\eqref{fxproblem} with time complexity $\mathcal{O}(1)$, which leads to a closed-form $f_n^*$.
\end{remark}

We then use ${f}_n^*$ in~\eqref{f*_finitehorizon} to establish a lower bound of $J_{n}^{\text{single}*}$.

\begin{theorem}  \label{zero-feedback lowerbound}
    In zero-feedback scenarios under the GAW traffic, for each $n\in \mathcal{N}$, we have
    \begin{align} \label{Jnsingle}
        J_{n}^{\text{single}*} \geq \frac{1 + \varepsilon_n}{2\rho_n (1-\varepsilon_n)} + \frac{1}{2}.
    \end{align}  
\end{theorem}
\begin{proof}
See Appendix~\ref{proofzerofeedbacklowerbound}.
\end{proof}

\begin{remark}  \label{lowerboundenergy}
For a generic point-to-point system with zero feedback and the GAW traffic,~\cite{Feng2021tcom} also established a lower bound under energy harvesting.
Different from ours, the proof therein considers a continuous-time model ideally assuming that every transmission occupies zero time, which allows
$\limsup_{T \rightarrow \infty} \mathbb{E}^{\bm{\pi}^\text{single}} [\frac{l_{n,k}}{U_n(T)}] = 0$
for any $k \geq 1$ to be used to significantly simplify the derivation.
But a little surprisingly, we find that the lower bound in~\cite{Feng2021tcom} is $\frac{1 + \varepsilon_n}{2(1-\varepsilon_n)}$, which is exactly $1/2$ less than ours when $\rho_n = 1$, although they are established under different system models. 
This is because, in~\cite{Feng2021tcom}, the energy constraint limits the AP to conduct at most $T$ transmissions, which coincides with the transmission rate constraint in $\eqref{problem1}$ when $\rho_n = 1$.
Note that the additional $1/2$ term stems from our discrete-time model where each transmission occupies one full slot.
\end{remark}

By~\eqref{Jpi} and~\eqref{Jnsingle}, for any $\bm{\pi} \in \bm{\Pi}$ with $q^{\bm{\pi}} \leq \rho$, we have
\begin{align}  \label{Jpibound}
    J^{\bm{\pi}} 
    \geq \sum_{n\in\mathcal{N}}\alpha_n J_{n}^{\text{single}*} 
    \geq \sum_{n \in \mathcal{N}} \alpha_n \bigg(\frac{1 + \varepsilon_n}{2\rho_n (1-\varepsilon_n)} + \frac{1}{2} \bigg),
\end{align}
where $\rho_n \in [0,\rho], \forall n \in \mathcal{N}$ and $\sum_{n \in \mathcal{N}} \rho_n \leq \rho$.
We minimize the right-hand side (RHS) of~\eqref{Jpibound} through solving
\begin{equation}  \label{problem_gaw_infinite}
    \begin{aligned}
    \min_{ \rho_n \in [0,\rho], \forall n \in \mathcal{N} } \quad &  \sum_{n \in \mathcal{N}} \alpha_n \bigg(\frac{1 + \varepsilon_n}{2\rho_n (1-\varepsilon_n)} + \frac{1}{2} \bigg),\\
    \text{s.t.} \quad\quad\quad & \sum_{n \in \mathcal{N}} \rho_n \leq \rho.
    \end{aligned}
\end{equation}
By the convex optimization theory~\cite{boyd2004convex}, we can obtain the solution of $\rho_n$ to~\eqref{problem_gaw_infinite} as
\begin{equation}  \label{bound_rhon}
\rho_n^* = \frac{\displaystyle \rho \sqrt{\alpha_n(1 + \varepsilon_n)/(1 - \varepsilon_n)}}{\displaystyle \sum_{n'\in\mathcal{N}} \sqrt{\alpha_{n'}(1 + \varepsilon_{n'})/(1 - \varepsilon_{n'})}}, \quad \forall n \in \mathcal{N}.
\end{equation}
Substituting~\eqref{bound_rhon} and $\sum_{n \in \mathcal{N}} \alpha_n = 1$ into the RHS of~\eqref{Jpibound}, we obtain the following lower bound.

\begin{theorem}  \label{infinitelowerbound}
In zero-feedback scenarios under the GAW traffic, for any $\bm{\pi} \in \bm{\Pi}$ with $q^{\bm{\pi}} \leq \rho$, we have 
\begin{equation}\label{f*_infinitehorizon_bound}
 J^{\bm{\pi}} \geq \frac{1}{2\rho} \left( \sum_{n \in \mathcal{N}} \sqrt{\alpha_n  \dfrac{1+\varepsilon_n}{1-\varepsilon_n}} \right)^2 + \frac {1} {2}. 
\end{equation}
\end{theorem}

\begin{remark}  \label{perfectlowerbound}
One may wonder about the impact of zero feedback on the lower bound derived in Theorem~\ref{infinitelowerbound}.
For perfect-feedback scenarios under the GAW traffic,~\cite{Ceran2021JSAC} provided a lower bound of infinite-horizon EWSAoI, which is given by 
\[\frac{1}{2\rho} \left(\sum_{n \in \mathcal{N}} \sqrt{\frac{\alpha_n}{1-\varepsilon_n}}\right)^2 + \frac{\rho}{2} \min_{n\in\mathcal{N}} \frac{\alpha_n\varepsilon_n}{(1-\varepsilon_n)} + \frac{1}{2}.
\]
Roughly speaking, the gap between these bounds increases linearly with $1+\varepsilon_n$ for each $n \in \mathcal{N}$ and decreases linearly with $\rho$, indicating that our bound considers a severer negative impact of zero feedback for more transmission errors or lower allowable maximum rate.
In a symmetric network where $\varepsilon_n = \varepsilon$ and $\alpha_n = 1/N$ for each $n\in\mathcal{N}$, the ratio between these bounds is exactly
\( 1 + \frac{\varepsilon(N^2 - \rho^2)}{N\rho(1-\varepsilon) + N^2 + \varepsilon \rho^2 } \).
When $\varepsilon_n = 0$ for each $n \in \mathcal{N}$, the two bounds coincide as $\frac{1}{2\rho} \left(\sum_{n\in\mathcal{N}} \sqrt{\alpha_n}\right)^2 + \frac{1}{2}$, indicating that our bound considers no impact of zero feedback without transmission errors.
Note that \cite{Kadota2021TMC} provides a slightly tighter lower bound than~\cite{Ceran2021JSAC} for $\rho = 1$ (which can be easily extended to $\rho \leq 1$), but lacks an explicit expression.
So, we will compare with it in Section~\ref{Zero Feedback: GAW Traffic}.
\end{remark}

\begin{remark}  \label{finitehorioznlowerbound}
   Due to the lack of an explicit expression of ${f}_n^*$, no lower bound of finite-horizon EWSAoI for zero-feedback scenarios under the GAW traffic is presented in~\cite{Munari2025TCOM}.
   To deal with this issue, we can obtain the following lower bound for such scenarios as a by-product, using~\eqref{f*_finitehorizon}, $U_n(T)\leq T\rho_n,\forall{n\in\mathcal{N}}$, and a similar idea as shown in~\eqref{Jpi}.
   \[
   \sum_{n \in \mathcal{N}}\alpha_n \left( \frac{T}{2} \frac{1+\varepsilon_n}{2 + (T\rho_n^* - 1)(1-\varepsilon_n)} + \frac{1}{2} \right),
   \]
    where $\rho_n^*$ can be obtained by solving
    \begin{align}
   \min_{ \rho_n \in [0,\rho], \forall n \in \mathcal{N}  } &\sum_{n \in \mathcal{N}}  \alpha_n \bigg(\frac{T}{2} \frac{1+\varepsilon_n}{2 + ( T\rho_n - 1)(1-\varepsilon_n)} + \frac{1}{2} \bigg), \notag \\
  \text{ s.t. } \ \ \quad &\sum_{n \in\mathcal{N}} \rho_n \leq \rho, \notag 
   \end{align}
   with the Karush–Kuhn–Tucker conditions~\cite{boyd2004convex}.
\end{remark}

\subsection{Optimal Policy for Many Cases}

Obviously, the solution of ${\hat{\bm{X}}_n}$ to problem~\eqref{fxproblem} \textit{cannot} be directly converted into an optimal policy, since the resulting transmission instants $\{l_{n,k}\}_{k \geq 1}$ may be non-integers and~\eqref{fxproblem} ignores the mutual restriction of different sources.
However, by the proof of Lemma~\ref{minimumfn} and Theorem~\ref{perfectlowerbound}, it is desirable for an optimal policy to spread out each source's transmissions as evenly as possible while maintaining the transmission rate as close to~\eqref{bound_rhon} as possible. 
For this consideration, we investigate when the lower bound in Theorem~\ref{perfectlowerbound} can be achieved, through designing an \textit{exact uniform scheduler} (EUS), which is a special class
of cyclic schedulers, with the individual transmission rates exactly equaling to~\eqref{bound_rhon}.
A policy is called an EUS~\cite{Li2021INFOCOM} if each source is scheduled periodically with periods that may be different across sources, i.e., 
$X_{n,k}  = X_{n,k'} $ for any $k,k'\in \mathbb{Z^+} \setminus \{1\}$ and each $n \in \mathcal{N}$, but $l_{n,k} \neq l_{n',k'}$ for any distinct $n, n' \in \mathcal{N}$ and any $k, k' \in \mathbb{Z^+}$.
For other cases, we defer the policy design issue to Section~\ref{Zero Feedback: Bernoulli Traffic}.

By~\eqref{eq: ERT} and~\eqref{bound_rhon}, it is easy to check that the lower bound in~\eqref{f*_infinitehorizon_bound} can be achieved if an EUS with the individual transmission rates exactly equaling to~\eqref{bound_rhon} is used.
However, EUSs have been constructed in~\cite{Li2021INFOCOM} only when any larger individual transmission rate can be divided by any smaller individual transmission rate.
Here, we establish an equivalent condition for the existence of EUSs and show that EUSs can be designed for many more cases than~\cite{Li2021INFOCOM}.

\begin{lemma}\label{lemma:tau}
Assume that $1/\rho_n^*$ in~\eqref{bound_rhon} is an integer for any $n\in \mathcal {N}$.
     An EUS with the individual transmission rates exactly equaling to~\eqref{bound_rhon} exists if and only if 
    \begin{equation} \label{eq:tau} 
        X_{n,1} \not\equiv X_{n',1} \mod \gcd(1/\rho_n^*,1/\rho_{n'}^*),
    \end{equation}
    for any distinct $n, n'\in \mathcal {N}$, where $\gcd$ denotes the greatest common divisor.
\end{lemma} 
\begin{proof}
See Appendix~\ref{proofoflemmatau}.
\end{proof}

By Lemma~\ref{lemma:tau}, we shall prove that $\{X_{n,1}\}_{n\in\mathcal{N}}$ satisfying~\eqref{eq:tau} can be found for many cases.
Our proof is based on an idea of jointly applying rate splitting and modular arithmetic, which can be described by a splitting tree.

A \emph{splitting tree} is a plane rooted tree equipped with weights on both its nodes and edges, defined recursively as follows:
\begin{itemize}
    \item the root has weight $1$; 
    \item If an internal node $\nu$ has weight $1/g_{\nu}$ and $p_{\nu}$ children, where $g_{\nu} \geq 1$ is an integer, $p_{\nu} \geq 2$ is a prime, then
each child of this node is assigned weight $1/(g_{\nu} p_{\nu})$, and the $p_{\nu}$ edges from this node to its children are assigned weights $0,g_{\nu},\ldots,(p_{\nu}-1)g_{\nu}$, respectively.
    \end{itemize}

Fig.~\ref{fig:splittingtree} illustrates an example of a splitting tree of $4$ internal nodes and $10$ leaves.
The root or an internal node can be seen as a \textit{virtual} source, while a leaf can be seen as an actual source, and its weight can be seen as a supportable individual rate.

\begin{figure}
\centering
\includegraphics[width = 3in]{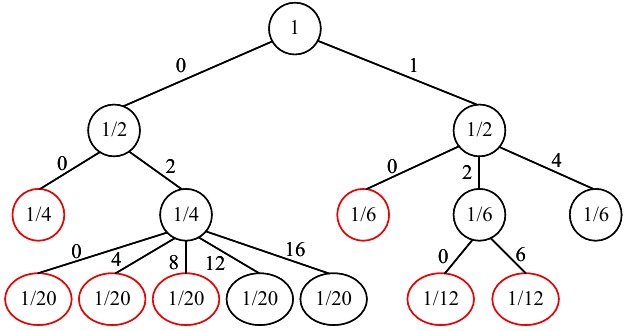}
    \caption{An example of a splitting tree.}  \label{fig:splittingtree} 
\end{figure}

\begin{lemma}\label{lemma:EUS2}
Assume that $1/\rho_n^*$ in~\eqref{bound_rhon} is an integer for any $n\in \mathcal {N}$.
An EUS with the individual transmission rates exactly equaling to~\eqref{bound_rhon} exists if
we can construct a splitting tree including $N$ leaves of weights $\{\rho^*_n\}_{n\in\mathcal{N}}$ and set the value of $X_{n,1}$ as the sum weights of the edges from the root to the leaf corresponding to weight $\rho^*_n$ for each $n\in\mathcal{N}$.
\end{lemma}
\begin{proof}
See Appendix~\ref{proofoflemmaEUS}.
\end{proof}

If an EUS with individual rates $\{1/\rho_n^*\}_{n\in\mathcal{N}}$ can be designed by a splitting tree, by Lemma~\ref{lemma:EUS2}, we need to examine at most 
\[
\prod_{n\in \mathcal{N}'}\zeta\Big(\frac{1/\rho_n^*}{\gcd(\{1/\rho_n^*\}_{n\in\mathcal{N}'})}\Big)
\] 
possible splitting trees to obtain the required EUS.
Here, $\mathcal{N}'$ denotes a subset of $\mathcal{N}$ containing exactly one index for each distinct rate and $\zeta(\cdot)$ denotes the number of permutations of a multiset consisting of prime factors of an integer. 
This method is significantly more efficient than exhaustively searching for suitable $\{X_{n,1}\}_{n\in\mathcal{N}}$.
For example, when $\rho^*_1=1/4$, $\rho^*_2=1/6$, $\rho^*_3=\rho^*_4=\rho^*_5=1/20$, $\rho^*_6=\rho^*_7=1/12$ under $N=7$, we have to construct at most $4$ splitting trees and can choose $7$ leaves in Fig.~\ref{fig:splittingtree} to set $X_{1,1}=0$, $X_{1,2}=1$, $X_{1,3}=2$, $X_{1,4}=6$, $X_{1,5}=10$, $X_{1,6}=3$, $X_{1,7}=9$.

\begin{remark}
    The case that any larger individual rate can be divided by any smaller individual rate~\cite{Li2021INFOCOM} turns out to be a special instance of those covered by Lemma~\ref{lemma:EUS2}, where the corresponding splitting tree satisfies the property that the branching numbers of all internal nodes divide one another.
\end{remark}

\begin{remark}  \label{AUSs}
Based on previously known EUSs, cyclic schedulers have been constructed in~\cite{Li2021INFOCOM,Liyanaarachchi2025TIT} for more general individual transmission rates, under various design goals. 
Thus, our results in Lemma~\ref{lemma:EUS2} would provide more foundations to construct more efficient cyclic schedulers.
However, this is beyond the scope of our paper and is not discussed here. 
\end{remark}

\section{Zero-Feedback Scenarios: Bernoulli Traffic}  \label{Zero Feedback: Bernoulli Traffic}
This section concentrates on zero-feedback scenarios under the Bernoulli traffic.
A common method is to model the problem~\eqref{problem1} as a CMDP, however, it is difficult to solve this CMDP through standard approaches, due to infinite states under zero feedback, intractable accessibility of these states, intractable optimal Lagrange
multiplier, and the curse of dimensionality for multiple sources.
To deal with this issue, we develop a low-complexity drift-plus-penalty (DPP) policy based on the theory of Lyapunov optimization, where the constraint in~\eqref{problem1} is enforced by transforming it into a queue stability constraint. 

\subsection{Hybrid Lyapunov Function}
\label{Hybrid Lyapunov Function}
Let $Q_t$ denote a virtual queue that characterizes the historical usage of the transmission budget, which evolves as
\begin{equation} \label{Qt_multiuser}
Q_{t+1} = \max\left[Q_t - \rho + \sum_{n \in \mathcal{N}} a_{n,t}, 0\right],
\end{equation}
where $Q_0 = 0$.
According to~\cite[Definition 2.3, Section 4.4]{Neely2010Stochastic}, the constraint in~\eqref{problem1} is satisfied as long as $Q_t$ is mean rate stable, i.e., $\limsup_{T \rightarrow \infty} \frac{\mathbb{E}\left[Q_T \right]}{T} = 0$.
On the other hand, $h_{n,t}$ for each $n \in \mathcal{N}$ can be seen as an actual queue that evolves as~\eqref{eq:Evolution_AoI}.
Based on these queues, we define a \textit{hybrid Lyapunov function} as follows:
\begin{equation} 
L_t \triangleq \frac{1}{2}V Q_t^2 + \sum_{n \in \mathcal{N}} \theta_n h_{n,t},
\label{eq:lyapunov_multi}
\end{equation}
where $\theta_n$ is a positive real number to represent the normalized queue weight of $h_{n,t}$ that will be further determined, and $V$ is a positive real number that represents the tradeoff between convergence and performance.
A larger $V$ would speed up the convergence to guarantee the constraint.

\begin{remark}  \label{Lyapunovfunction}
One may wonder why we use a hybrid Lyapunov function rather than a traditional quadratic function~\cite{Kadota2018ACM,Luo2025Tcom} or a linear function~\cite{Kadota2021TMC,ZhaoINFOCOM2025}.
When computing the DPP function~\eqref{Dtht}, we will find that the quadratic on $Q_t$ leads to a dominant term $Q_t^2$, which will be essential to prove Theorem~\ref{thm:DPP}.  
Also, we will find that the linearity on the queue $h_{n,t}$ can be used to offset the impact of the penalty $h_{n,t+1}$ by choosing an appropriate $\theta_n$, as shown in Appendix~\ref{proofqueuestability}, 
which will be essential to not only prove Theorem~\ref{thm:DPP} but also design a threshold policy as in~\eqref{select_threshold} without the need of any truncation on $h_{n,t}$, unlike other Lyapunov optimization based approaches~\cite{Zakeri2024TWC,Fountoulakis2023TCOM}.
So, this hybrid form brings a ``magic'' to the subsequent analysis.
\end{remark}

\subsection{DPP Function}\label{DPP Function}

Let $\bm{o}_{n,t} \triangleq \left\{\{d_{n,t'}\}_{t' = 0}^{t}, \{a_{n,t'}\}_{t' = 0}^{t-1}, \{v_{n,t'}\}_{t' = 0}^{t} \right \}$ denote all historical information related to source $n$ before scheduling at slot $t$.
Let $\bm{o}_t \triangleq \left\{ \{\bm{o}_{n,t}\}_{n \in \mathcal{N}}, Q_t \right\}$.
The Lyapunov drift $\Delta_t^{ \bm{\pi} } $, defined as the conditional expected change in~\eqref{eq:lyapunov_multi} over one slot under a general policy $\bm{\pi} \in \bm{ \Pi^{\mathrm{}}}$, is given by
\begin{equation} \label{Deltat}
\Delta_t^{ \bm{\pi}} \triangleq \mathbb{E}^{ \bm{\pi} } \left [L_{t+1} - L_t \mid \bm{o}_t\right],
\end{equation}
where the expectation $\mathbb{E}^{ \bm{\pi} }[\cdot]$ is taken with respect to scheduling actions made in reaction to $\bm{o}_t$ under $\bm{\pi}$.

Letting the penalty in slot $t$ be the expected AoI at slot $t+1$, we define the DPP function under a policy $\bm{\pi} \in \bm{ \Pi^{\mathrm{}}}$ as
\begin{align}  
    {C}_t^{ \bm{\pi}} 
    & \triangleq \Delta_t^{ \bm{\pi}} + \sum_{n \in \mathcal{N}} \alpha_n \mathbb{E}^{ \bm{\pi} }[ h_{n,t+1} \mid \bm{o}_t] \notag \\
    & = \mathbb{E}^{ \bm{\pi} } \Big [ \frac{V}{2} (Q_{t+1}^2 - Q_t^2)+ \sum_{n \in \mathcal{N}} \theta_n ( h_{n,t+1} - h_{n,t}) \notag \\ 
    & \quad \quad \quad + \sum_{n \in \mathcal{N}} \alpha_n h_{n,t+1} \mid \bm{o}_t\Big].  \label{Dthtpi}
\end{align}
Define $\hat{h}_{n,t}  \triangleq \mathbb{E}^{ \bm{\pi} }\left[h_{n,t} \mid \bm{o}_{t} \right]$.
Substituting~\eqref{eq:Evolution_AoI} into~\eqref{Dthtpi} and then manipulating the resulting expression, we obtain
\begin{align}  \label{Dtht}
    {C}_t^{ \bm{\pi}}& 
     =\ \mathbb{E}^{ \bm{\pi} } \bigg [ \frac{V}{2} (Q_{t+1}^2 - Q_t^2) + \sum_{n \in \mathcal{N}} \alpha_n  h_{n,t} \notag \\ 
   & + \sum_{n \in \mathcal{N}} ( \theta_n + \alpha_n )\big( a_{n,t}(1-\varepsilon_n)( w_{n,t} - h_{n,t}) + 1\big) \mid \bm{o}_t \bigg] \notag \\
  & =\ \mathbb{E}^{ \bm{\pi} } \bigg [ \frac{V}{2} (Q_{t+1}^2 \! - \! Q_t^2) \mid \bm{o}_t\bigg]\! +\! 1 + \sum_{n \in \mathcal{N}} ( \theta_n + \alpha_n  \hat{h}_{n,t} ) \notag \\
     &+ \sum_{n \in \mathcal{N}} ( \theta_n\! + \! \alpha_n )(1\! -\! \varepsilon_n)( w_{n,t} \! -\! \hat{h}_{n,t}) \mathbb{E}^{ \bm{\pi} } \big [ a_{n,t} \mid \bm{o}_t\big],  
\end{align}
where the last equality is due to that $a_{n,t}$ is independent of $w_{n,t}$, $\hat{h}_{n,t}$ given $\bm{o}_{t}$ and $\sum_{n \in \mathcal{N}} \alpha_n = 1$.

By~\eqref{Qt_multiuser}, we have
$Q_{t+1}^2 - Q_t^2 \leq \rho^2 + 1 + 2Q_t\big(\sum_{n \in \mathcal{N}} a_{n,t} - \rho\big)$
due to $\sum_{n \in \mathcal{N}} a_{n,t} \leq 1$.
Substituting it into~\eqref{Dtht}, we obtain 
\begin{align}  \label{Dtbound_multiuser}
 {C}_t^{ \bm{\pi}} \leq & \frac{V(\rho^2 + 1) }{2} + 1 + VQ_t ( \mathbb{E}^{ \bm{\pi} } [ \sum_{n \in \mathcal{N}} a_{n,t} \mid \bm{o}_t] - \rho) \notag \\
 & 
+ \sum_{n \in \mathcal{N}} ( \theta_n + \alpha_n )(1-\varepsilon_n) \notag  \big( w_{n,t} - \hat{h}_{n,t} \big) \mathbb{E}^{ \bm{\pi} } \big [ a_{n,t} \mid \bm{o}_t\big] \notag \\
&+ \sum_{n \in \mathcal{N}} \big( \theta_n + \alpha_n \hat{h}_{n,t}  \big).
\end{align}
By opportunistically minimizing the RHS of~\eqref{Dtbound_multiuser}, we get a policy, denoted by $\bm{\pi}^{\text{DPP}}$, in the following form:
\begin{align} 
    \min_{\substack{ a_{n,t} \in \{0,1\}\\ \forall{n}\in\mathcal{N} } }  & \sum_{n \in \mathcal{N}} \big( VQ_t \!+\!  ( \theta_n\! +\! \alpha_n )(1\!-\!\varepsilon_n)( w_{n,t}\! -\! \hat{h}_{n,t}) \big) a_{n,t},\notag \\
    \text{s.t.} \quad &\sum_{n \in \mathcal{N}} a_{n,t} \leq 1. \label{problem}
\end{align}
In the following, we prove that, if $\{\theta_n\}_{n\in\mathcal{N}}$ is chosen suitably, $\bm{\pi}^{\text{DPP}}$ is able to satisfy the transmission rate constraint, enjoy a simple structure, and enjoy a performance guarantee.

\subsection{Queue Stability} \label{Queue Stability}

By choosing suitable $\{\theta_n\}_{n\in\mathcal{N}}$, we show that the virtual queue $Q_t$ is mean rate stable under $\bm{\pi}^{\text{DPP}}$. 
The proof employs a feasible stationary randomized policy $\bm{\pi}^{\text{ran}}$, which requires the AP to schedule source $n$ with probability $0 \leq \eta_n \leq \rho$ at each slot, as a reference policy to derive the results. 
Note that $\sum_{n \in \mathcal{N}} \eta_n = \rho$ under $\bm{\pi}^{\text{ran}}$.

\begin{theorem} \label{thm:DPP}

Under the policy $\bm{\pi}^{\text{DPP}}$ with 
\begin{align}  \label{theta_n}
   \theta_n = \frac{\alpha_n(1 - ( 1 - \varepsilon_n) \eta_n )}{(1-\varepsilon_n) \eta_n}, \quad \forall n\in \mathcal{N},
\end{align}
$Q_t$ is mean rate stable, i.e., 
$\limsup_{T \rightarrow \infty} \frac{\mathbb{E}^{ \bm{\pi}^{\text{DPP}} } \left[Q_T \right]}{T} = 0$.
\end{theorem}

\begin{proof}
See Appendix~\ref{proofqueuestability}.
\end{proof}

\subsection{Performance Guarantee}
Still relying on $\bm{\pi}^{\text{ran}}$ and carefully determined $\{\theta_n\}_{n\in\mathcal{N}}$ in~\eqref{theta_n}, we obtain a performance guarantee of $\bm{\pi}^{\text{DPP}}$.
\begin{theorem}   \label{thm:upp}
Under the policy $\bm{\pi}^{\text{DPP}}$ with~\eqref{theta_n},
we have
\begin{align} \label{eq:upper}
J^{\bm{\pi}^{\text{DPP}}}
    \leq \frac{V(\rho^2 + 1)}{2}  + \sum_{n \in \mathcal{N}} \alpha_n  \left( \frac{1 }{(1-\varepsilon_n) \eta_n^*} +  \frac{1-\lambda_n}{\lambda_n} \right),
\end{align}
where
\begin{equation}  \label{optimaletan}
\eta_n^* = \frac{\rho \sqrt{ \alpha_n/(1-\varepsilon_n)}}{\sum\limits_{n'\in\mathcal{N}} \sqrt{\alpha_{n'}/(1-\varepsilon_{n'})}}, \quad \forall n \in \mathcal{N}.
\end{equation}
\end{theorem}
\begin{proof}
See Appendix~\ref{proofoftheoremthmupp}.
\end{proof}

\begin{remark}  \label{lyapunovupperbound}
It is shown in~\cite{Kadota2021TMC,Ji2024TCOM} that 
\[
J^{\bm{\pi}^{\text{ran*}}}=\sum_{n \in \mathcal{N}} \alpha_n  \left( \frac{1 }{(1-\varepsilon_n) \eta_n^*} +  \frac{1-\lambda_n}{\lambda_n} \right),
\] 
where $\bm{\pi}^{\text{ran*}}$ denotes an optimal stationary randomized policy.
Comparing $J^{\bm{\pi}^{\text{ran*}}}$ and~\eqref{eq:upper}, we obtain
\begin{equation}
    J^{\bm{\pi}^{\text{DPP}}} \leq  \frac{V(\rho^2 + 1)}{2} + J^{\bm{\pi}^{\text{ran*}}},
\end{equation}
which implies that the upper bound of $J^{\bm{\pi}^{\text{DPP}}}$ can be arbitrarily close to $J^{\bm{\pi}^{\text{ran*}}}$ even with zero feedback.
When $V=0$ and $\rho=1$, this upper bound reduces to that in~\cite{Kadota2021TMC}.
Establishing a tighter upper bound is challenging since we \textit{cannot} use properties such as renewal intervals to simplify the analysis.
The superior performance of $\bm{\pi}^{\text{DPP}}$ will be validated via simulations.
\end{remark}

\subsection{Policy Structure}\label{Policy Structure}
 
Substituting~\eqref{optimaletan} and~\eqref{theta_n} into problem~\eqref{problem}, we obtain the following optimization problem:
\begin{equation}
\begin{aligned} \label{problem_simple}
    \min_{a_{n,t} \in \{0,1\}, \forall{n\in\mathcal{N}}} \quad &\sum_{n \in \mathcal{N}} \Big( VQ_t + \alpha_{n}\frac{w_{n,t} - \hat{h}_{n,t}} {\eta_n^*} \Big) a_{n,t} ,\\
    \text{s.t.}\ \ \quad\quad\quad&\sum_{n \in \mathcal{N}} a_{n,t} \leq 1. 
\end{aligned}
\end{equation}
We can obtain $\bm{\pi}^{\text{DPP}}$ in the following threshold-based form:
\begin{align}  \label{select_threshold}
   a_{n,t} \in 
    \begin{cases}
       1,   & \text{if } VQ_t \leq \alpha_{n}\frac{\hat{h}_{n,t} - w_{n,t}} {\eta_n^*}, m_t=n, \\
       0,   & \text{otherwise},
    \end{cases}
\end{align}
where $m_t \in\text{arg min}_{n\in\mathcal{N}} \left( VQ_t + \alpha_{n}\frac{w_{n,t} - \hat{h}_{n,t}} {\eta_n^*} \right)$ 
and $\hat{h}_{n,t}$ can be computed recursively as~\cite{Ji2024TCOM,Munari2025TCOM}
\begin{align}  \label{evolution_hath}
    \hat{h}_{n,t+1} = 
    \begin{cases}
        \varepsilon_n (\hat{h}_{n,t} + 1) + (1 - \varepsilon_n)(w_{n,t} + 1),  & \text{if } a_{n,t} = 1, \\
        \hat{h}_{n,t} + 1,  & \text{if } a_{n,t} = 0.
    \end{cases}
\end{align}

\begin{remark} \label{threshold structure}
We observe from~\eqref{Qt_multiuser} and \eqref{evolution_hath} that if the AP schedules source $n$ in slot $t$, $Q_t$ increases by $1 - \rho$ while $\hat{h}_{n,t}$ changes by $(1-\varepsilon_n)( w_{n,t} - \hat{h}_{n,t} ) + 1$; otherwise, $Q_t$ decreases with $\rho$ while $\hat{h}_{n,t}$ increases by $1$. 
So, under the threshold structure of $\bm{\pi}^{\text{DPP}}$ as stated in~\eqref{select_threshold}, $Q_t$ and $\hat{h}_{n,t}$ are controlled as two mutually counterbalancing processes. 
Note that, as $\lambda_n$ increases, $w_{n,t}$ is more frequently reduced to zero, which makes $\hat{h}_{n,t}$ decrease more quickly.
So, the term $\hat{h}_{n,t} - w_{n,t}$, called conditional expected age-based weight~\cite{Kadota2021TMC}, is used in~\eqref{select_threshold} together with $Q_t$ to provide an index that reasonably measures the transmission preference.
\end{remark}

\section{Extension to Imperfect-Feedback Scenarios}\label{Extension to Imperfect-Feedback Scenarios}

This section extends the DPP policy proposed in Section~\ref{Zero Feedback: Bernoulli Traffic} to consider general imperfect feedback, which incorporates arbitrary feedback delays and feedback errors.
Note that Sections~\ref{Hybrid Lyapunov Function}--\ref{Policy Structure} are applicable for general imperfect feedback.
since they are independent of the feedback details.
However, computing $\hat{h}_{n,t}$ under general imperfect feedback depends on feedback delays, errors, and mechanisms, which would lead to more complicated observations and calculations than that for zero feedback as shown in~\eqref{evolution_hath}.

Let $\bm{b}_{n,t} \triangleq [b_{n,t}(1), b_{n,t}(2), \dots ]$ denote the probability distribution of $h_{n,t}$, where $b_{n,t}(h) \triangleq \Pr(h_{n,t} = h \mid \bm{b}_{n,t-1}, w_{n,t-1}, \bm{o}_{n,t} )$ denotes the conditional probability of $h_{n,t} = h$ given $\bm{b}_{n,t-1}, w_{n,t-1}, \bm{o}_{n,t}$.
We have $\hat{h}_{n,t} = \sum_{h\in\mathbb{Z^+}}  h b_{n,t}(h)$.
As the feedback delay is $D_n$, at the beginning of each slot $t + 1$, the AP is able to update $\bm{b}_{n,t - D_n + 1}$ based on $\bm{b}_{n,t-D_n},w_{n,t-D_n},a_{n,t-D_n},v_{n,t+1}$ and then update $\{\bm{b}_{n,t'}\}_{t' = t - D_n +2}^{t+1}$ based on the updated $\bm{b}_{n,t - D_n + 1}$, $\{a_{n,t'}\}_{t' = t - D_n +1}^{t}$, $\{w_{n,t'}\}_{t' = t - D_n +1}^{t}$.
We introduce this procedure as follows.

\emph{\underline{Step 1:}} By the Bayes’ rule, for each $h' \in \mathbb{Z^+}$, at the beginning of slot $t + 1$, the AP can update $b_{n,t-D_n+1}(h')$ as~\eqref{ht_belief},
\begin{figure*}
  \begin{align} \label{ht_belief}
    b_{n,t - D_n+1}(h') 
    & \triangleq  \Pr(h_{n,t-D_n+1} = h' \mid \bm{b}_{n,t-D_n}, w_{n,t-D_n}, a_{n,t-D_n}, v_{n,t+1} ) \notag \\
    & = \frac{\sum_{h\in\mathbb{Z^+}}b_{n,t-D_n}(h) \xi_{n,t-D_n+1}(h' \mid h,w,a)\beta_{n,t + 1}(v \mid h', h, w, a) }{\sum_{h''\in\mathbb{Z^+}}\sum_{h\in\mathbb{Z^+}}b_{n,t-D_n}(h) \xi_{n,t-D_n+1}(h'' \mid h,w,a)\beta_{n,t+1}(v \mid h'',h,w,a)},
  \end{align}
\hrulefill %
\end{figure*}
where
\begin{align} \label{xi}
 &\xi_{n,t-D_n+1}(h' \mid h,w,a) \triangleq  \Pr(h_{n,t-D_n+1} = h' \mid \notag \\ &\quad\quad\quad\quad\quad  h_{n,t-D_n} = h, w_{n,t-D_n} = w,a_{n,t-D_n} = a) \notag \\
&  = \begin{cases}
      \varepsilon_n,   & \text{if }  a = 1, h' = h + 1, \\
      1 - \varepsilon_n, & \text{if } a = 1, h' = w + 1,   \\
      1, & \text{if } a = 0, h' = h + 1, \\
      0, & \text{otherwise},
  \end{cases} 
\end{align}
denotes the state transition function that is applicable for all feedback mechanisms, and
    \begin{align} \label{function_observation}
    & \beta_{n,t+1}(v \mid h',h,w,a)  \triangleq \Pr(v_{n,t+1} = v \mid \hat{h}_{t - D_n + 1} = h',\notag \\ & \quad \quad \quad  \hat{h}_{t-D_n} = h, w_{n,t-D_n} = w, a_{n,t-D_n} = a) ,
\end{align}
denotes the observation function that obviously depends on the implemented feedback mechanism.

\emph{\underline{Step 2:}} Due to the feedback delay $D_n$, the transmission outcomes of source $n$ of slots $t - D_n + 1,\ldots,t$ are completely unobservable to the AP at the beginning of slot $t+1$.
Thus, the updating of $\{\bm{b}_{n,t'}\}_{t' = t - D_n +2}^{t+1}$ can be reduced to updating $\{\hat{h}_{n,t'}\}_{t' = t - D_n +2}^{t+1}$ by~\eqref{evolution_hath}.

In the following, we derive $\beta_{n,t+1}(v \mid h',h,w,a)$ and $\hat{h}_{n,t-D_n+1}$ under different feedback mechanisms.

\subsection{ACKs Mechanism}
\label{Calculation of hath with ACKs}

Under the ACKs mechanism, we have
\begin{align} \label{function_observation_ack}
    & \beta_{n,t+1}^{\text{ACKs}}(v \mid h',h,w,a) \notag \\
    & = 
    \begin{cases}
        1,  & \text{if } a = 0, v = 0, h' = h + 1, \\
            &  \text{or } a = 1, v = 0, h' = h + 1,\\
        \sigma_n,    &  \text{if } a = 1, v = 0, h' = w + 1,\\
        1-\sigma_n, &  \text{if } a = 1, v = 1, h' = w + 1,\\
        0, & \text{otherwise}.
    \end{cases}
\end{align}
Substituting~\eqref{xi} and~\eqref{function_observation_ack} into~\eqref{ht_belief}, the AP can update $\hat{h}_{n, t - D_n+1}$ as~\eqref{hDNack}.
\begin{figure*}
\begin{align}\label{hDNack}
    \hat{h}_{n,t-D_n+1} 
    & = \frac{\sum_{h'\in\mathbb{Z^+}} h' \sum_{h\in\mathbb{Z^+}}b_{n,t-D_n}(h) \xi_{n,t-D_n+1}(h' \mid h,w,a)\beta_{n,t+1}^{\text{ACKs}}(v \mid h', h, w, a) }{\sum_{h''\in\mathbb{Z^+}}\sum_{h\in\mathbb{Z^+}}b_{n,t-D_n}(h) \xi_{n,t-D_n+1}(h'' \mid h,w,a)\beta_{n,t+1}^{\text{ACKs}}(v \mid h'',h,w,a)}\notag \\
    & = 
    \begin{cases}
        \frac{(w + 1)(1-\varepsilon_n) (1-\sigma_n) \sum_{h\in\mathbb{Z^+}}b_{n,t-D_n}(h)}{(1-\varepsilon_n) (1-\sigma_n)\sum_{h\in\mathbb{Z^+}}b_{n,t-D_n}(h) }, & \text{if } a = 1, v = 1,  \\
        \frac{\varepsilon_n  \sum_{h\in\mathbb{Z^+}} (h + 1) b_{n,t-D_n}(h) + (1-\varepsilon_n)\sigma_n(w + 1)\sum_{h\in\mathbb{Z^+}}b_{n,t-D_n}(h) }{( 1 - (1-\varepsilon_n)(1-\sigma_n) )\sum_{h\in\mathbb{Z^+}}b_{n,t-D_n}(h)} , & \text{if } a = 1, v = 0,  \\
        \frac{\sum_{h\in\mathbb{Z^+}} (h+1) b_{n,t-D_n}(h)}{\sum_{h\in\mathbb{Z^+}}b_{n,t-D_n}(h)}, & \text{if } a = 0.
    \end{cases}
\end{align}

\end{figure*}
Considering that $\sum_{h\in\mathbb{Z^+}}h b_{n,t-D_n}(h) = \hat{h}_{n,t-D_n}$ and $\sum_{h\in\mathbb{Z^+}}b_{n,t-D_n}(h) = 1$,
we can compute $\hat{h}_{n,t-D_n+1}$ exactly without relying on the belief of $\hat{h}_{n,t-D_n}$. 
\begin{align} \label{htD_ACK}
    & \hat{h}_{n,t-D_n+1}= \notag \\ &
    \begin{cases}
        w_{n,t-D_n} + 1,   & \text{if } a_{n,t-D_n} = 1, v_{n,t} = 1, \\
        \frac{\varepsilon_n ( \hat{h}_{n,t-D_n} + 1 ) }{1 - (1-\varepsilon_n)(1-\sigma_n)} \\+ \frac{(1 - \varepsilon_n)\sigma_n ( w_{n,t-D_n} + 1 )}{1 - (1-\varepsilon_n)(1-\sigma_n)},                   & \text{if } a_{n,t-D_n} = 1, v_{n,t} = 0,  \\
\hat{h}_{n,t-D_n} + 1,                   & \text{if } a_{n,t-D_n} = 0.
    \end{cases}
\end{align}

\subsection{ACKs/NACKs Mechanism}

Under the ACKs/NACKs mechanism, we have
\begin{align} \label{beta_ackandnack}
    &\beta_{n,t+1}^{\text{ACKs/NACKs}}(v \mid h',h,w,a) \notag \\
    & = 
    \begin{cases}
        1,  & \text{if } a = 0, v = 0, h' = h + 1, \\
        1 - \sigma_n    &  \text{if } a = 1, v = -1, h' = h + 1,\\
                       &  \text{or } a = 1, v = 1, h' = w + 1,\\
        \sigma_n    &  \text{if } a = 1, v = 0, h' = h + 1,\\
                   &  \text{or } a = 1, v = 0, h' = w + 1,\\       
        0, & \text{otherwise}.
    \end{cases}
\end{align}
Substituting~\eqref{xi} and~\eqref{beta_ackandnack} into~\eqref{ht_belief}, the AP can update $\hat{h}_{n, t - D_n+1}$ as~\eqref{htDn_ACKNACK}.
\begin{figure*}
    \begin{align} \label{htDn_ACKNACK}
    \hat{h}_{n,t-D_n+1} 
    & = \frac{\sum_{h'\in\mathbb{Z^+}} h' \sum_{h\in\mathbb{Z^+}}b_{n,t-D_n}(h) \xi_{n,t-D_n+1}(h' \mid h,w,a)\beta_{n,t+1}^{\text{ACKs/NACKs}}(v \mid h', h, w, a) }{\sum_{h''\in\mathbb{Z^+}}\sum_{h\in\mathbb{Z^+}}b_{n,t-D_n}(h) \xi_{n,t-D_n+1}(h'' \mid h,w,a)\beta_{n,t+1}^{\text{ACKs/NACKs}}(v \mid h'',h,w,a)}\notag \\
    & = 
    \begin{cases}
        \frac{(w + 1)(1-\varepsilon_n) (1-\sigma_n) \sum_{h\in\mathbb{Z^+}}b_{n,t-D_n}(h)}{(1-\varepsilon_n) (1-\sigma_n)\sum_{h\in\mathbb{Z^+}}b_{n,t-D_n}(h) }, & \text{if } a = 1, v = 1,  \\
        \frac{\varepsilon_n\sigma_n \sum_{h\in\mathbb{Z^+}} (h + 1) b_{n,t-D_n}(h) + (1-\varepsilon_n)\sigma_n(w + 1)\sum_{h\in\mathbb{Z^+}}b_{n,t-D_n}(h) }{\sigma_n\sum_{h\in\mathbb{Z^+}}b_{n,t-D_n}(h)} , & \text{if } a = 1, v = 0,  \\
        \frac{ \varepsilon_n(1 - \sigma_n) \sum_{h\in\mathbb{Z^+}} (h+1) b_{n,t-D_n}(h)}{\varepsilon_n (1 - \sigma_n) \sum_{h\in\mathbb{Z^+}}b_{n,t-D_n}(h)}, & \text{if } a = 1, v = -1, \\
        \frac{\sum_{h\in\mathbb{Z^+}} (h+1) b_{n,t-D_n}(h)}{\sum_{h\in\mathbb{Z^+}}b_{n,t-D_n}(h)}, & \text{if } a = 0.
    \end{cases}
\end{align}
\end{figure*}
Considering that $\sum_{h\in\mathbb{Z^+}}h b_{n,t-D_n}(h) = \hat{h}_{n,t-D_n}$ and $\sum_{h\in\mathbb{Z^+}}b_{n,t-D_n}(h) = 1$,
we can compute $\hat{h}_{n,t-D_n+1}$ exactly without relying on the belief of $\hat{h}_{n,t-D_n}$. 
\begin{align} \label{htD_ACKNACK}
    &\hat{h}_{n,t-D_n+1} = \notag \\
    &\begin{cases}
        w_{n,t-D_n} + 1,   & \text{if } a_{n,t-D_n} = 1, v_{n,t+1} = 1,  \\
        {\varepsilon_n ( \hat{h}_{n,t-D_n} + 1 ) } +\\ (1 - \varepsilon_n) ( w_{n,t-D_n} + 1 ),                   & \text{if } a_{n,t-D_n} = 1, v_{n,t+1} = 0,  \\
        \hat{h}_{n,t-D_n} + 1,                   & \text{if } a_{n,t-D_n} = 1, v_{n,t+1} = -1, \\
        &\text{or } a_{n,t-D_n} = 0.
    \end{cases}
\end{align}
Note that~\eqref{htD_ACKNACK} is independent of $\sigma_n$ due to the belief normalization in~\eqref{htDn_ACKNACK}.
This is different from~\eqref{htD_ACK} for ACKs mechanism.

\subsection{Complexity Simplification}

By~\eqref{evolution_hath},~\eqref{htD_ACK},~\eqref{htD_ACKNACK}, the computational complexity of $\bm{\pi}^{\text{DPP}}$ under imperfect feedback is $\mathcal{O}\big(N + \sum_n D_n\mathbf{1}_{D_n < \infty, \sigma_n <1}\big)$ where $\mathbf{1}_{[\cdot]}$ denotes the indicator function, for both the ACKs and ACKs/NACKs mechanisms, which is undesirable for large $N$ or $D_n$.
To reduce such complexity, we use the form in~\cite{Ji2024TCOM} to simplify the updating of $\hat{h}_{n,t+1}$ at the beginning of slot $t+1$ into~\eqref{evolutionhn_simple},
\begin{figure*}
\begin{align} \label{evolutionhn_simple}
    \hat{h}_{n,t+1}\! =\! & \ \big(1\! - \!a_{n,t}(1\! -\! \varepsilon_n) \big)\big( \hat{h}_{n,t} \!+\!\varepsilon_n^{ \sum_{j = t - D_n + 1}^{t-1} a_{n,j} } \big( \hat{h}_{n,t-D_n+1}\! -\! \big(a_{n,t-D_n}(1\!-\!\varepsilon_n)(w_{n,t-D_n}\!-\!\hat{h}_{n,t-D_n})\! +\! \hat{h}_{n,t-D_n}\! +\! 1 \big) \big)\! - \! w_{n,t} \big)\notag \\ & + w_{n,t} + 1,
\end{align}
\hrulefill %
\end{figure*}
where $\hat{h}_{n,t-D_n+1}$ is calculated using~\eqref{htD_ACK} and~\eqref{htD_ACKNACK} for the ACKs and ACKs/NACKs mechanisms, respectively, which account for feedback errors that are not considered in~\cite{Ji2024TCOM}. 
By~\eqref{htD_ACK},~\eqref{htD_ACKNACK},~\eqref{evolutionhn_simple}, we can obtain the online computational complexity of $\bm{\pi}^{\text{DPP}}$ as listed in Table~\ref{complexity}.
We also compare it with that of other policies, which indicates that, interestingly, introducing arbitrary feedback delays and arbitrary feedback errors into $\bm{\pi}^{\text{DPP}}$ leads to no extra computational complexity compared to assuming zero feedback.
Note that cyclic schedulers~\cite{Li2021INFOCOM,Liyanaarachchi2025TIT} have the lowest complexity, but have the most limited applicability.  
 
\begin{table}[ht] 
\centering  
\footnotesize
\caption{Comparison of online computational complexity.}
\label{complexity}
\begin{tabular}{@{\hspace{1pt}}c@{\hspace{1pt}}c@{\hspace{1pt}}c}
\toprule
{}    & Complexity   & Applicability \\
\midrule
{\makecell{Our EUS and other \\  cyclic schedulers~\cite{Li2021INFOCOM,Liyanaarachchi2025TIT}}}     &    $\mathcal{O}(1)$          & \makecell{ zero feedback \\ GAW traffic}   \\
\midrule
\makecell{$\bm{\pi}^{\text{DPP}}$\\ (Section~\ref{Zero Feedback: Bernoulli Traffic})}     &    $\mathcal{O}(N)$               & \makecell{ zero feedback \\ Bernoulli traffic}\\
\midrule
\makecell{$\bm{\pi}^{\text{DPP}}$ \\ (Section~\ref{Extension to Imperfect-Feedback Scenarios})}     &    $\mathcal{O}(N)$               & \makecell{  delayed/erroneous feedback \\ Bernoulli traffic}  \\
\midrule
\makecell{Lagrange-cost \\greedy~\cite{Ji2024TCOM} }    &    $\mathcal{O}(N)$        & \makecell{  delayed  feedback \\ Bernoulli traffic} \\
\bottomrule
\end{tabular}
\end{table}

\begin{remark}
    For each user $n$, when $\sigma_n$, $D_n$ take specific values, both~\eqref{htD_ACK} and~\eqref{htD_ACKNACK} can be reduced to~\eqref{evolution_hath} for zero feedback~\cite{Munari2025TCOM}, while~\eqref{htD_ACKNACK} can be reduced to that for delayed feedback~\cite{Ji2024TCOM} and erroneous feedback~\cite{PyttelICC2024} under ACKs/NACKs. 
    Note that~\cite{PyttelICC2024} maintained a belief of $h_{n,t}$, which is unnecessary for designing $\bm{\pi}^{\text{DPP}}$.
\end{remark}

\begin{remark}  \label{reducetomax-weight}
    By setting $V = 0,\rho = 1$, the policy $\bm{\pi}^{\text{DPP}}$ in~\eqref{select_threshold} can be reduced to the form of the max-weight policy~\cite{Kadota2021TMC} for perfect feedback, except that $\bm{\pi}^{\text{DPP}}$ requires more complicated calculation of $\hat{h}_{n,t}$ due to imperfect feedback.
\end{remark}

\section{Numerical Results}
\label{Numerical Results}
In this section, we verify our findings through simulations.
We consider our proposed lower bound in Theorem~\ref{infinitelowerbound} for zero feedback, our constructed EUSs by Lemma~\ref{lemma:EUS2}, our proposed $\bm{\pi}^{\text{DPP}}$ in~\eqref{select_threshold} under both ACKs and ACKs/NACKs, the lower bounds in~\cite{Kadota2021TMC} for perfect feedback, the cyclic scheduler~\cite{Liyanaarachchi2025TIT}, and the Lagrange-cost greedy policy with ACKs/NACKs~\cite{Ji2024TCOM} for comparison.
We comply with the system model specified in Section~\ref{sec:SystemModel} to set up the numerical experiments and shall change the network configuration over a broad range to examine the impact of various factors.
Each result is an average of $10$ independent numerical experiments, each of which lasts for $10^6$ slots. 
We set $V = 1$ in the threshold form~\eqref{select_threshold} of the proposed $\bm{\pi}^{\text{DPP}}$.
Note that, for zero feedback, the EWSAoI of $\bm{\pi}^{\text{DPP}}$ is independent of the adopted feedback mechanism.
Also note that, our EUS and the cyclic scheduler~\cite{Liyanaarachchi2025TIT} are only designed for zero feedback and GAW traffic.

\subsection{Zero Feedback: GAW Traffic}  \label{Zero Feedback: GAW Traffic}

Table~\ref{tab: weightA} compares the EWSAoI of different schemes for $N = 4$ sources under zero feedback and the GAW traffic.
We set $\varepsilon_n = 0.1$ for each $n \in \mathcal{N}$ and $\alpha_1:\alpha_2:\alpha_3:\alpha_4 = 1:4:9:36$.
Note that we only can construct EUSs under some specific configurations by Lemma~\ref{lemma:EUS2}, which is not covered in~\cite{Li2021INFOCOM}.
It shows that the EWSAoI under all the schemes always decreases as $\rho$ increases, since a higher $\rho$ permits more transmission opportunities.
As expected, we observe that our EUS always achieves the lower bound in Theorem~\ref{infinitelowerbound}.
Not too surprisingly, we observe that $\bm{\pi}^{\text{DPP}}$ always approaches the lower bound in Theorem~\ref{infinitelowerbound} and slightly outperforms the cyclic scheduler~\cite{Liyanaarachchi2025TIT} with $1.71\%$-$3.66\%$ improvement.
This is because the threshold structure of $\bm{\pi}^{\text{DPP}}$ strikes a good balance between the long-term constraint and the EWSAoI, which leads to almost evenly distributed transmissions for each source, while the cyclic scheduler~\cite{Liyanaarachchi2025TIT} pursues a short-term constraint at the cost of degrading the transmission evenness.
We further observe that $\bm{\pi}^{\text{DPP}}$ outperforms the Lagrange-cost greedy policy~\cite{Ji2024TCOM} with $9.65\%$--$18.90\%$ improvement.
This is because the latter lacks an optimal Lagrange multiplier and requires a real-time constraint, which jointly cause worse transmission evenness.
We also note that the gap between the two bounds decreases almost linearly with $\rho$, which is consistent with the theoretical comparison in Remark~\ref{perfectlowerbound}.

\begin{table}[htbp]
    \centering
    \scriptsize
    \caption{EWSAoI in zero-feedback scenarios with GAW traffic for 4 sources with asymmetric priorities.}
    \label{tab: weightA}   
    \begin{tabular}{@{\hspace{1pt}}c@{\hspace{5pt}}c@{\hspace{5pt}}c@{\hspace{5pt}}c@{\hspace{5pt}}c@{\hspace{5pt}}c@{\hspace{5pt}}c}
    \toprule
    {$\rho$}& {\makecell{Lower bound \\perfect feedback\\~\cite{Kadota2021TMC}}}   & {\makecell{Lower bound \\ zero feedback\\ (Thm.~3)}} & {\makecell{Our EUS}} & {\makecell{Cyclic \\Scheduler\\~\cite{Liyanaarachchi2025TIT}}} & {\makecell{Our $\bm{\pi}^{\text{DPP}}$}} &{\makecell{Lagrange-cost \\greedy~\cite{Ji2024TCOM}}}  \\
    \midrule  
    0.1 &16.50  & 18.10 & 18.10 & 18.72 & 18.10  & 22.34 \\
    {1/6} &10.10  & 11.06 & 11.06 & 11.42 & 11.06  & 13.58  \\
    0.2  &8.50  & 9.30  & 9.30  & 9.58  & 9.30 & 11.40  \\
    {1/4}&6.90  & 7.54  & 7.54  & 7.79  & 7.54  & 9.25    \\
    0.3  &5.83   & 6.37  & {--}  & 6.56  & 6.38  & 7.88 \\ %
    0.5  &3.70  & 4.02  & 4.02  & 4.13  & 4.02 & 4.87   \\
    0.8  &2.50  & 2.70  & {--}  & 2.83  & 2.73  & 3.02  \\ %
    1.0  &2.10  & 2.26  & {--}  & 2.40  & 2.32  & 2.68 \\ %
    \bottomrule
    \end{tabular}
\end{table}

\subsection{Imperfect Feedback: Bernoulli Traffic}
Figs.~\ref{fig:verus weightstep}--\ref{fig:verus Dsigma_imperfectfeedbackBernoulli} compare the EWSAoI of different schemes in imperfect-feedback scenarios with Bernoulli traffic for $N = 12$ sources that are equally divided into four groups, each of which has the same priority.
We set $D_n = D$, $\sigma_n =\sigma$, $\varepsilon_n =\varepsilon$, and $\lambda_n = \lambda$ for each $n \in \mathcal{N}$.
To examine the impact of priorities, in Fig.~\ref{fig:verus weightstep}, we set $\lambda = 0.5$, $\varepsilon = 0.2$, $D = 10$, $\sigma = 0.2$, and set the priority ratio among $4$ groups as $1 : 1+r: 1+2r: 1+3r$ in Fig.~\ref{Arithmetic Sequence} and $1 : r: r^2 : r^3$ in Fig.~\ref{Geometric Sequence}, for different $r$.
To examine the impact of other factors, in Fig.~\ref{fig:verus Dsigma_imperfectfeedbackBernoulli}, we set the priority ratio among $4$ groups as $1:4:7:10$ and consider various values of $\lambda$, $\rho$, $\varepsilon$, $D$, $\sigma$.

\begin{figure}[!ht]
    \centering
    \subfigure[$1 : 1+r: 1+2r: 1+3r$.]{
    \label{Arithmetic Sequence}
    \includegraphics[width = 1.5in]{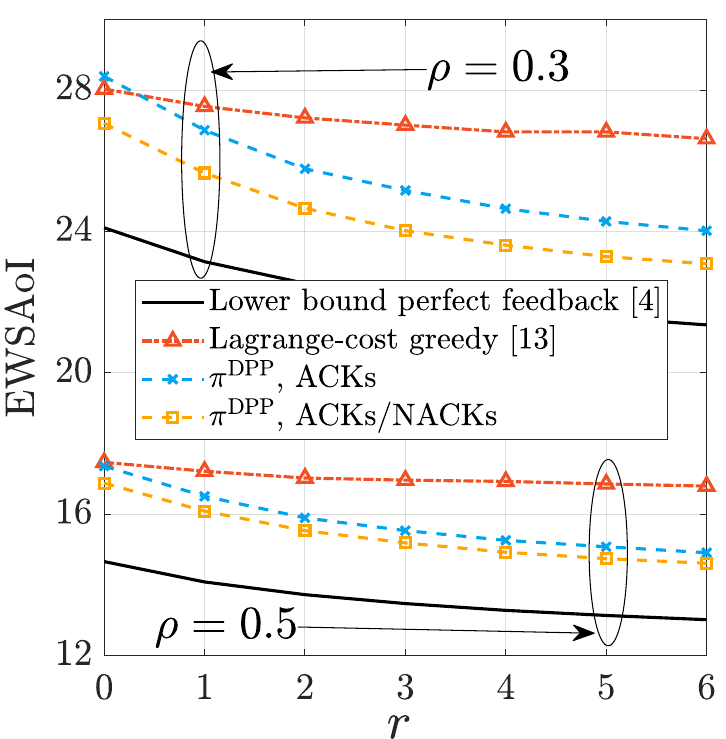}
}
    \subfigure[$1 : r: r^2 : r^3$.]{
    \label{Geometric Sequence}
	\includegraphics[width = 1.5in]{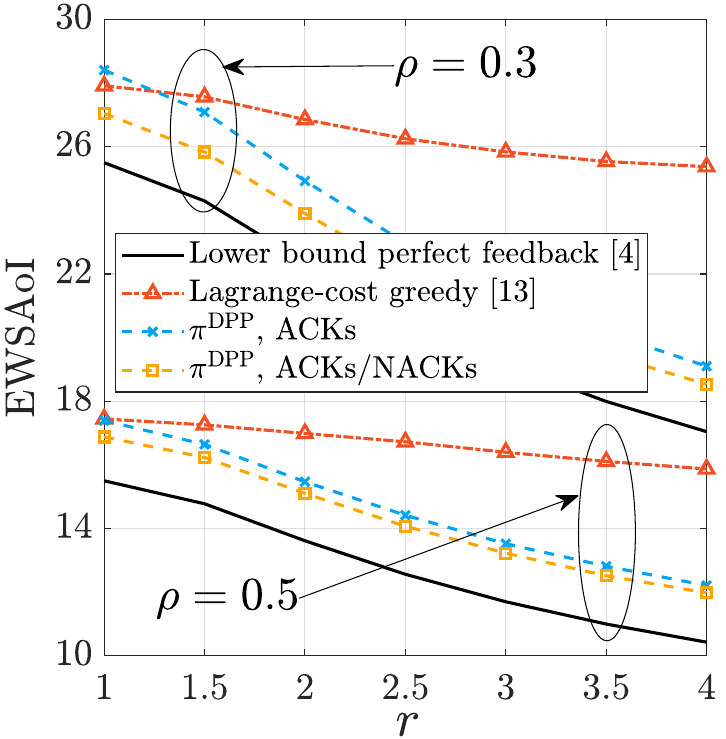}}
	\caption{EWSAoI vs. priority ratio in imperfect-feedback scenarios with Bernoulli traffic for $12$ sources.}
	\label{fig:verus weightstep}
\end{figure}

\subsubsection{Performance Comparison}
As evidenced by Figs.~\ref{fig:verus weightstep}--\ref{fig:verus Dsigma_imperfectfeedbackBernoulli}, the proposed $\bm{\pi}^{\text{DPP}}$ with ACKs/NACKs performs best in all the cases owing to its well-designed threshold structure~\eqref{select_threshold} as explained in Remark~\ref{threshold structure} and more accurate estimation of actual instantaneous AoI as shown in~\eqref{htD_ACKNACK}, 
while $\bm{\pi}^{\text{DPP}}$ with ACKs performs second since the ACKs mechanism provides less information than ACKs/NACKs.
The Lagrange-cost greedy policy~\cite{Ji2024TCOM} performs worst in almost all the cases due to the reasons as explained in Section~\ref{Zero Feedback: GAW Traffic}.
Specifically, the proposed $\bm{\pi}^{\text{DPP}}$ with ACKs/NACKs enjoys $0\%$--$5.83\%$ improvement over the proposed $\bm{\pi}^{\text{DPP}}$ with ACKs and $0.61\%$--$26.97\%$ improvement over the Lagrange-cost greedy policy~\cite{Ji2024TCOM}.
It is interesting to note that the $\bm{\pi}^{\text{DPP}}$ with ACKs generally outperforms the Lagrange-cost greedy policy~\cite{Ji2024TCOM} because of the well-designed threshold structure, although the former uses a weaker feedback mechanism.
As expected, the EWSAoI of all the schemes is closer to the lower bound for perfect feedback~\cite{Kadota2021TMC} under smaller $\sigma$, $D$, $\varepsilon$, or $\rho$ due to a weaker impact of imperfect feedback.

\subsubsection{Impact of priorities}
Fig.~\ref{fig:verus weightstep} shows that the EWSAoI of all the schemes decreases with $r$.
This is because it is easier to serve multiple sources when priorities are more unequally apportioned among the sources.
We observe that the gap between $\bm{\pi}^{\text{DPP}}$ with ACKs/NACKs and the Lagrange-cost greedy policy~\cite{Ji2024TCOM} becomes larger as $r$ increases.
This is because the weakness of the latter as introduced in Section~\ref{Zero Feedback: GAW Traffic} becomes more severe under more unequal priorities. 
We also observe that the gap between $\bm{\pi}^{\text{DPP}}$ with ACKs/NACKs and ACKs slightly decreases as $r$ increases.
This is because, under more unequal priorities, the priorities play a more important role in decision-making than the estimation accuracy of the instantaneous AoI.

\begin{figure*}[!ht]
	\centering
    \subfigure[ {\scriptsize $\lambda = 0.5$, $\sigma = 0.3$, $\rho = 0.5$.} ]{
    \centering
    \label{Imperfect_D}
    \includegraphics[width = 1.55in]{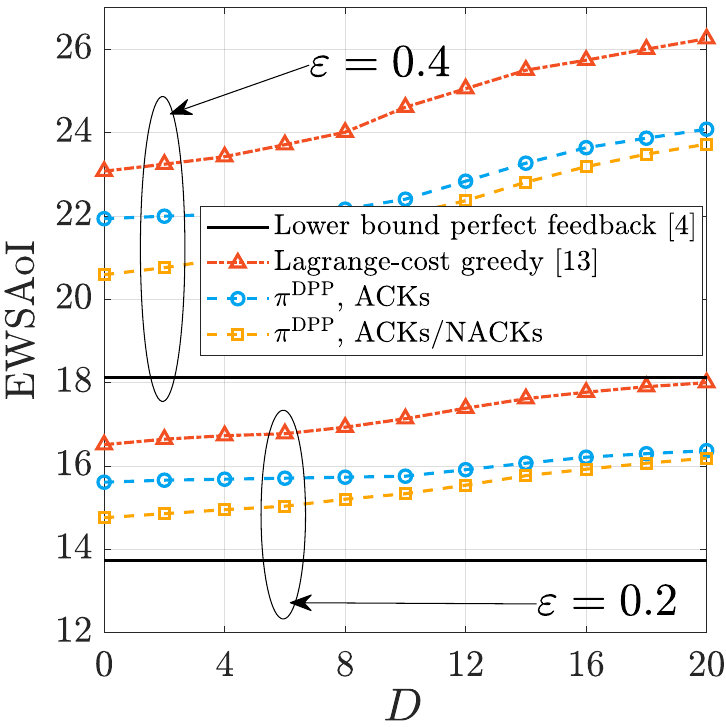}
}
    \subfigure[{\scriptsize $\lambda = 0.5$, $D = 10$, $\varepsilon = 0.2$.}]{
    \centering
    \label{Imperfect_sigma}
	\includegraphics[width = 1.55in]{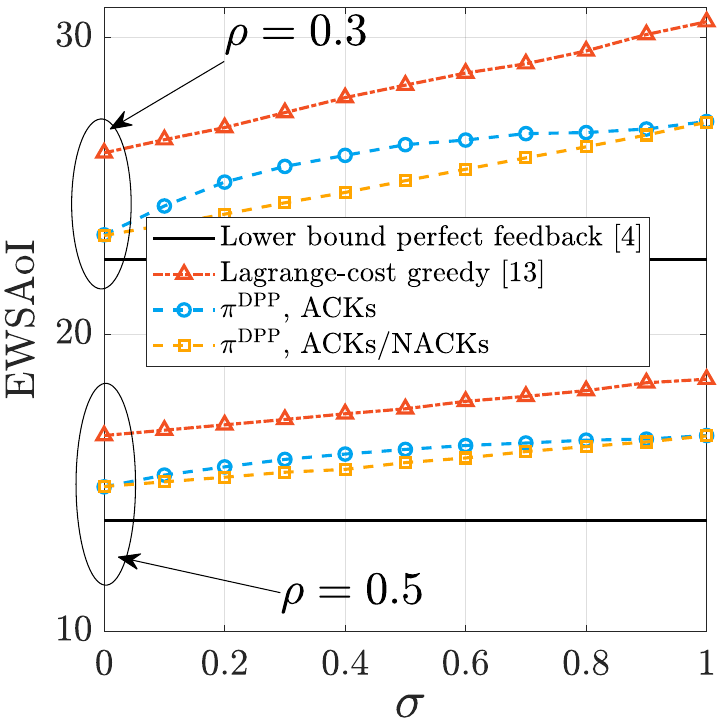}} 
    \subfigure[{\scriptsize $\lambda = 0.5$, $D = 10$, $\sigma = 0.2$.}]{
   \label{Imperfect_rho}
    \includegraphics[width = 1.55in]{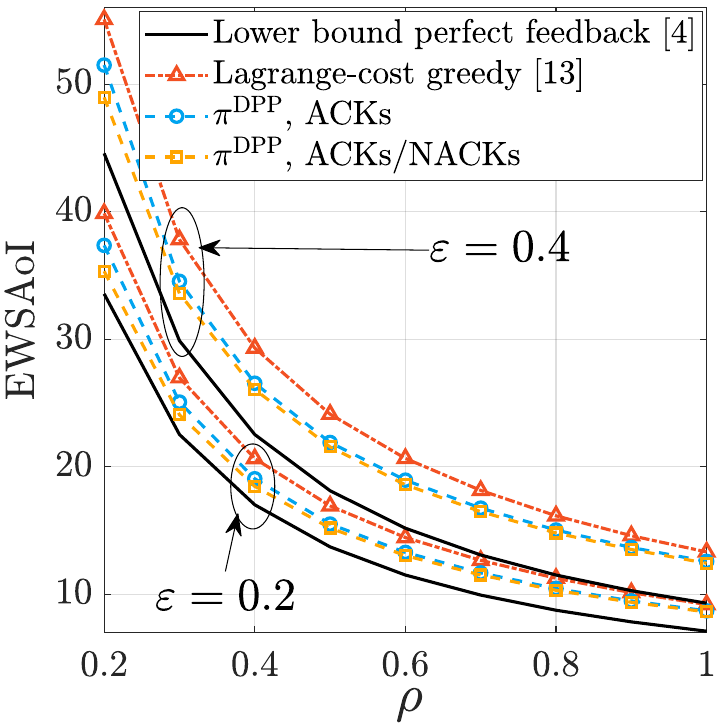}}
    \subfigure[{\scriptsize  $D = 10$, $\rho = 0.5$, $\varepsilon = 0.2$.}]{
   \label{Imperfect_lambda}
	\includegraphics[width = 1.55in]{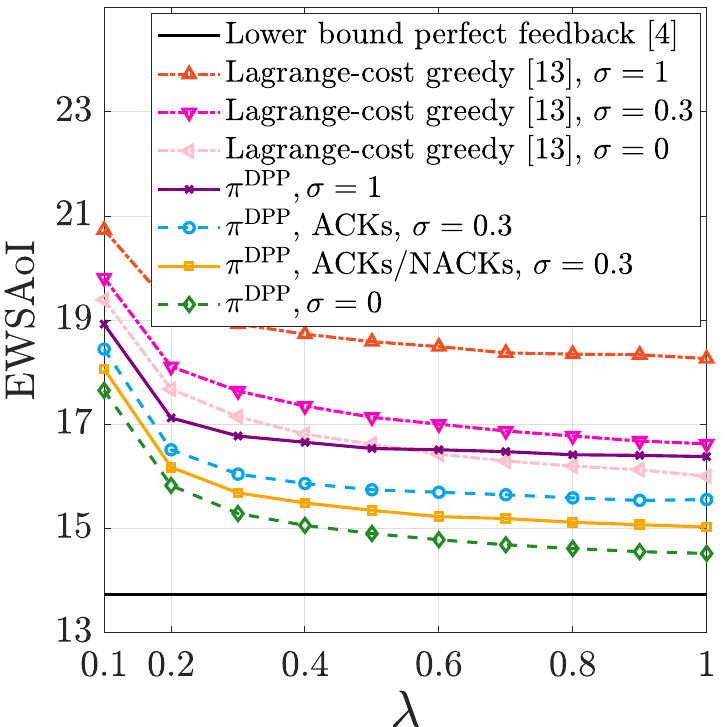}}
	\caption{EWSAoI in imperfect-feedback scenarios with Bernoulli traffic for $12$ sources.}
	\label{fig:verus Dsigma_imperfectfeedbackBernoulli}
\end{figure*}

\subsubsection{Impact of feedback}
Figs.~\ref{Imperfect_D}--\ref{Imperfect_sigma} show that the EWSAoI of all the schemes increases with $D$ and $\sigma$.
This is because the available feedback information decreases as $D$ or $\sigma$ increases, leading to less reasonable scheduling decisions.
Figs.~\ref{Imperfect_D}--\ref{Imperfect_sigma} also show that the performance gap between $\bm{\pi}^{\text{DPP}}$ with ACKs/NACKs and the Lagrange-cost greedy policy~\cite{Ji2024TCOM} decreases as $D$ or $\sigma$ increases.
This is because the efficacy of reasonably leveraging real-time observations for scheduling becomes weaker under less feedback information.
We further observe that the gap between $\bm{\pi}^{\text{DPP}}$ with ACKs/NACKs and ACKs becomes smaller as $D$ increases.
This is because a more delayed NACK becomes less helpful to estimate the instantaneous AoI as shown in~\eqref{htD_ACKNACK}. 
We also observe that the gap between $\bm{\pi}^{\text{DPP}}$ with ACKs/NACKs and ACKs first becomes larger and then becomes smaller as $\sigma$ increases.
This is because, by comparing~\eqref{htD_ACK} and~\eqref{htD_ACKNACK}, the ACKs/NACKs mechanism is more helpful to estimate the AoI than the ACK mechanism when $\sigma$ is farther away from $0$ or $1$, but this benefit diminishes when $\sigma$ approaches $0$ or $1$.

\subsubsection{Impact of transmission rate constraint and transmission errors}
Fig.~\ref{Imperfect_rho} shows that the EWSAoI of all the schemes decreases with $\rho$ due to more transmission opportunities.
We observe that, as $\rho$ increases, the performance gap between all the schemes with the same $\varepsilon$ decreases.
This is because increasing $\rho$ reduces the benefit of not only more reasonable scheduling but also a stronger feedback mechanism. 
Figs.~\ref{Imperfect_D} and~\ref{Imperfect_rho} show that the EWSAoI under all the schemes increases with $\varepsilon$ due to less successful transmissions.
We also observe that the performance gap between all the schemes decreases as $\varepsilon$ increases.
This is because the benefit of reasonable scheduling becomes less pronounced under more transmission errors.

\subsubsection{Impact of network traffic}
Fig.~\ref{Imperfect_lambda} shows that the EWSAoI of all the schemes generally decreases as $\lambda$ increases, due to the delivery of fresher updates.
We observe that the performance gap between all the schemes becomes smaller as $\lambda$ decreases.
This is because the benefit of more reasonable scheduling becomes weaker for more sporadic packets.

\begin{figure}[!ht]
    \centering
    \subfigure[asymmetric $\lambda_n$ and $\sigma_n$.]{
    \label{asyn_rho}
    \includegraphics[width = 1.52in]{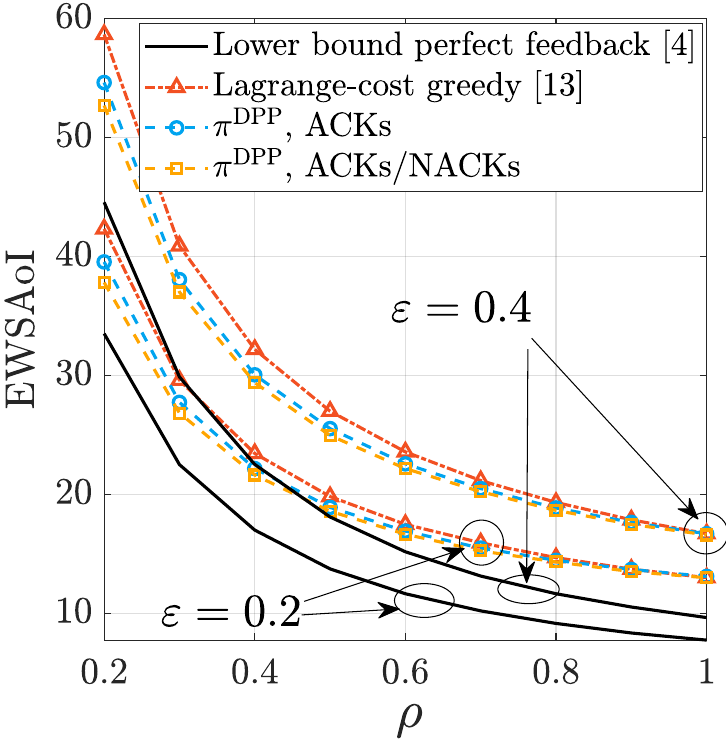}
}
    \subfigure[asymmetric $\varepsilon_n$ and $D_n$.]{
    \label{asyn_N}
	\includegraphics[width = 1.52in]{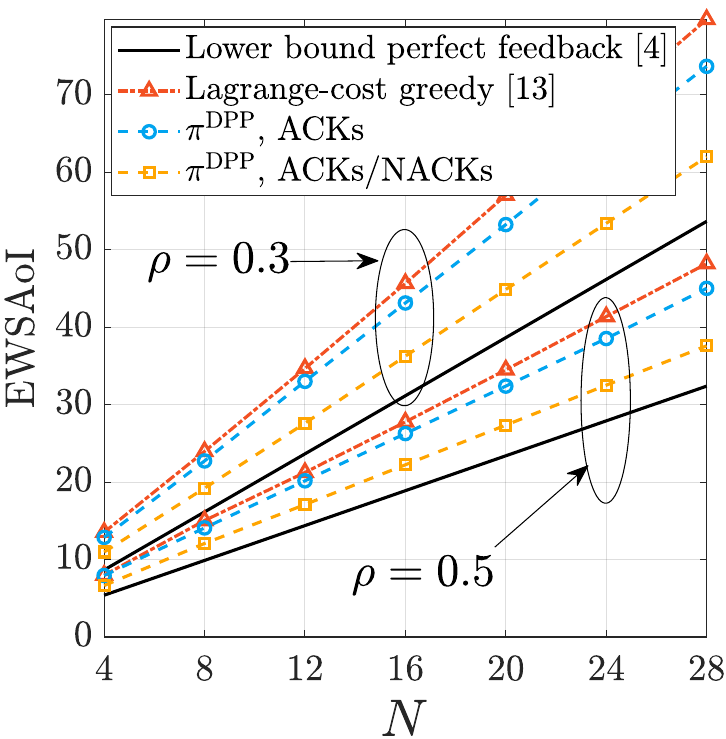}}
	\caption{EWSAoI in imperfect-feedback scenarios with Bernoulli traffic with more asymmetric parameters.}
	\label{fig: asyn}
\end{figure}

With the same priority setting in Fig.~\ref{fig:verus Dsigma_imperfectfeedbackBernoulli}, Fig.~\ref{fig: asyn} compares the EWSAoI of different schemes in more asymmetric cases.
In Fig.~\ref{asyn_rho}, we set $N=12$, $D_n = 10$, $\varepsilon_n = \varepsilon$, $\lambda_n = \frac{N-n+1}{2N}$, $\sigma_n = \frac{n-1}{2N}$ for each $n \in \mathcal{N}$.
In Fig.~\ref{asyn_N}, we set $\lambda_n = 0.5$, $\sigma_n = 0.2$, $\varepsilon_n = \frac{N-n+1}{2N}$ for each $n \in \mathcal{N}$, $D_{n} = 5$ for $n\in\{1,\dots,N/2\}$, and $D_{n} = 10$ for $n\in\{N/2+1,\dots,N\}$.
We observe that $\bm{\pi}^{\text{DPP}}$ with ACKs/NACKs enjoys $0.89\%$--$16.11\%$ improvement over $\bm{\pi}^{\text{DPP}}$ with ACKs and $0.14\%$--$22.24\%$ improvement over the Lagrange-cost greedy policy~\cite{Ji2024TCOM}.
We still observe that $\bm{\pi}^{\text{DPP}}$ with ACKs outperforms the Lagrange-cost greedy policy~\cite{Ji2024TCOM} in almost all the cases.
Fig.~\ref{asyn_N} shows that the EWSAoI increases linearly with $N$.
This is because increasing $N$ reduces the transmission opportunities available to each source under the fixed transmission rate constraint.
We also observe that the performance gap between all the schemes becomes more notable as $N$ increases.
This is due to the fact that the transmission opportunities for each source decrease as $N$ increases, which makes both reasonable scheduling and an effective feedback mechanism more important.

\section{Conclusions} \label{Conclusions}
This paper has investigated scheduling design for optimizing the infinite-horizon EWSAoI in downlink systems with imperfect feedback and transmission rate constraints.  
For zero feedback under the GWA traffic, a closed-form lower bound of achievable EWSAoI and a policy that achieves this bound in many cases have been proposed.
For zero feedback under the Bernoulli traffic, a threshold-based policy has been developed based on the theory of Lyapunov optimization and a closed-form performance guarantee has been provided.
Furthermore, this policy has been extended to consider general imperfect feedback under both ACKs and ACKs/NACKs mechanisms.
Numerical results demonstrated the performance advantage of the proposed policies in various scenarios. 
These findings reveal the joint impact of imperfect feedback and transmission rate constraints, which can be used to guide the preference to establish suitable feedback channels, choose suitable feedback mechanisms, and assign suitable network resources. 
A possible direction for future work would be to investigate decentralized uplink schemes with imperfect feedback.

\appendices

\section{Proof of Lemma~\ref{minimumfn}}  \label{prooflemmaminimumf}
Define the inner product \(\langle \bm{i}, \bm{j} \rangle \triangleq \bm{i}^{\text{T}} \bm{A}_n \bm{j}\) and the induced norm \(\|\bm{i}\| \triangleq \sqrt{\langle \bm{i}, \bm{i} \rangle} = \sqrt{\bm{i}^{\text{T}} \bm{A}_n \bm{i}}\).
By the Cauchy-Schwarz inequality
$\langle \bm{i}, \bm{j} \rangle^2 \leq \|\bm{i}\|^2 \|\bm{j}\|^2$ \cite[Theorem 6.16]{Linear},
we have
\begin{align}  \label{cauchy}
\langle \bm{A}_n^{-1} \bm{e}, {\hat{\bm{X}}_n} \rangle^2 & \leq \|\bm{A}_n^{-1} \bm{e}\|^2 \cdot \|{\hat{\bm{X}}_n}\|^2 \notag \\ 
\implies
    {\hat{\bm{X}}}_n^{\text{T}} \bm{A}_n {\hat{\bm{X}}_n} & \geq \frac{(\bm{e}^{\text{T}} {\hat{\bm{X}}_n})^2}{\bm{e}^{\text{T}} \bm{A}_n^{-1} \bm{e}}.
\end{align}
The equality in~\eqref{cauchy} holds if and only if \({\hat{\bm{X}}_n}\) is proportional to \(\bm{A}_n^{-1} \bm{e}\), i.e., \({\hat{\bm{X}}_n} = c \bm{A}_n^{-1} \bm{e}\) for some constant \(c\).
Substituting~\eqref{cauchy} into the objective function of problem~\eqref{fxproblem}, we obtain
\begin{align}
    {f}_n({\hat{\bm{X}}_n}) \geq \frac{T}{2}\frac{1}{\bm{e}^\text{T} \bm{A}_n^{-1}  \bm{e}} + \frac{1}{2}.  
\end{align}
Thus, the solution to problem~\eqref{fxproblem} can be obtained when the equality in~\eqref{cauchy} holds and is given by
\begin{align}   \label{X*}
    {\hat{\bm{X}}}_n^{*} = c \bm{A}_n^{-1} \bm{e} 
     = \frac{T \bm{A}_n^{-1} \bm{e}}{\bm{e}^{\text{T}} \bm{A}_n^{-1} \bm{e}},
\end{align}
where the last equality is due to \(\bm{e}^{\text{T}} {\hat{\bm{X}}_n} = T\), specifically
\begin{align}
\bm{e}^{\text{T}} (c \bm{A}_n^{-1} \bm{e}) = T \implies c = \frac{T}{\bm{e}^{\text{T}} \bm{A}_n^{-1} \bm{e}}.
\end{align}
Meanwhile, the minimum of $\hat{f}({\hat{\bm{X}}_n})$ is given by
\begin{align}  \label{f*}
    \hat{f}_n^* = \frac{T}{2} \frac{1}{\bm{e}^{\text{T}} \bm{A}_n^{-1} \bm{e}} + \frac{1}{2}.
\end{align}
The terms $\bm{A}_n^{-1} \bm{e}$ and \(\bm{e}^{\text{T}} \bm{A}_n^{-1} \bm{e}\) can be calculated as
\begin{align}  \label{A-1}
\bm{A}_n^{-1} \bm{e} & = \left( \frac{1}{1+\varepsilon_n}, \frac{1-\varepsilon_n}{1+\varepsilon_n}, \dots, \frac{1-\varepsilon_n}{1+\varepsilon_n}, \frac{1}{1+\varepsilon_n} \right)^{\text{T}}, \\
\bm{e}^{\text{T}} \bm{A}_n^{-1} \bm{e} 
&= \frac{2 + (U_n(T)-1)(1-\varepsilon_n)}{1+\varepsilon_n} \label{1A1}.
\end{align}
Substituting~\eqref{A-1} and~\eqref{1A1} into~\eqref{X*}, we obtain
\begin{align}  \label{x*_finitehorizon}
\hat{X}_{n,k}^* =
\begin{cases}
    \frac{T}{2 + (U_n(T)-1)(1-\varepsilon_n)}, &\text{if } k = 1, U_n(T) + 1, \\
    \frac{T(1-\varepsilon_n)}{2 + (U_n(T)-1)(1-\varepsilon_n)},   &\text{if } k = 2, \dots, U_n(T).
\end{cases}
\end{align}
Note that $\hat{X}_{n,k}^* \in [0,T-1]$ for each $k \in \{1,2,\ldots,U_n(T)+1\}$.
So, substituting~\eqref{1A1} into~\eqref{f*} yields~\eqref{f*_finitehorizon}.

\section{Proof of Theorem~\ref{zero-feedback lowerbound}} \label{proofzerofeedbacklowerbound}
From~\eqref{f*_finitehorizon}, we have
\begin{align}
& \mathbb{E}^{\bm{\pi}^\text{single}}\left[\frac{R_n(T)}{T}\mid \{l_{n,k}\}_{k=1}^{U_n(T) } \right] \notag \\
& \geq \frac{T}{2} \frac{1+\varepsilon_n}{2 + (U_n(T) - 1)(1-\varepsilon_n)} + \frac{1}{2}.   \label{eq: ERTUT}
\end{align}
Taking the expectation of both sides of~\eqref{eq: ERTUT} with respect to the distribution of $\{l_{n,k}\}_{k=1}^{U_n(T)}$ under all possible $U_n(T)$, and letting $T$ approach $\infty$, we obtain
\begin{align}   \label{ERT_infinite}
 &\limsup_{T \to \infty} \mathbb{E}^{\bm{\pi}^\text{single}}\left[\frac{R_n(T)}{T} \right] \notag \\
& \geq \limsup_{T \to \infty} \mathbb{E}^{\bm{\pi}^\text{single}}\left[\frac{T}{2}\frac{1+\varepsilon_n}{2 + (U_n(T) - 1)(1-\varepsilon_n)} + \frac{1}{2} \right]  \notag \\
& 
 \geq \frac{1 + \varepsilon_n}{2\rho_n(1-\varepsilon_n)} + \frac{1}{2},
\end{align}
due to $\limsup_{T \to \infty}\mathbb{E}^{\bm{\pi}^\text{single}}\left[ \frac{U_n(T)}{T}\right] \leq \rho_n$ under~\eqref{problem_n}.

\section{Proof of Lemma~\ref{lemma:tau}} \label{proofoflemmatau}
Under an EUS with the specified individual transmission rates, for each $n \in \mathcal{N}$, the scheduling action of source $n$ starts at $l_{n,1}=X_{n,1}$ and exhibits periodicity with period $1/{\rho_n^*}$ (or a multiple thereof).
So, such an EUS exists if and only if 
\begin{equation}\label{eq:tau_equiv}
    X_{n,1}+k\cdot1/\rho_{n}^* \neq X_{n',1}+k'\cdot1/\rho_{n'}^*
\end{equation}
for any distinct $n,n'\in \mathcal{N}$ and any $k,k'\in\mathbb{N}$.
The result follows since the conditions in~\eqref{eq:tau},~\eqref{eq:tau_equiv} are equivalent.

\section{Proof of Lemma~\ref{lemma:EUS2}} \label{proofoflemmaEUS}

Due to space limitations, we give a proof sketch here.
For given $n\neq n'$, if $\rho^*_n=\rho^*_{n'}$, the two corresponding nodes in the specified splitting tree share a common parent node. 
Following the tree structure, we have that $X_{n,1}\neq X_{n',1}$ and both are less than $1/\rho^*_n$, which is clearly the $\gcd$ of $1/\rho^*_n$ and $1/\rho^*_{n'}$. 
Hence, \eqref{eq:tau} holds.
This argument can be iteratively generalized to any two leaves by the structure of a splitting tree.
So, by Lemma~\ref{lemma:tau}, we can design a required EUS.

\section{Proof of Theorem~\ref{thm:DPP}}  \label{proofqueuestability}

Implementing $\bm{\pi}^{\text{DPP}}$ into~\eqref{Dtbound_multiuser} yields:
\begin{align}  
   {C}_t^{\bm{\pi}^{\text{DPP}}} 
   & \leq  \frac{V(\rho^2 + 1)}{2} + 1 + VQ_t \Big( \mathbb{E}^{ \bm{\pi}^{\text{DPP}} } \left[ \sum_{n \in \mathcal{N}} a_{n,t} \mid \bm{o}_t\right] - \rho \Big) \notag \\
&
+ \sum_{n \in \mathcal{N}} ( \theta_n + \alpha_n )(1-\varepsilon_n)( w_{n,t} - \hat{h}_{n,t}) \mathbb{E}^{ \bm{\pi}^{\text{DPP}} } \big [ a_{n,t} \mid \bm{o}_t\big] 
\notag \\ & + \sum_{n \in \mathcal{N}} \big( \theta_n + \alpha_n \hat{h}_{n,t} \big)
   \notag \\
   & \leq \frac{V(\rho^2 + 1)}{2} + 1 + \sum_{n \in \mathcal{N}} \big( \theta_n + \alpha_n \hat{h}_{n,t} \big)
   \notag \\
   & + \sum_{n \in \mathcal{N}} ( \theta_n + \alpha_n )(1-\varepsilon_n)( w_{n,t} - \hat{h}_{n,t}) \eta_n,
   \label{ref_multiuser}
\end{align}
where~\eqref{ref_multiuser} due to that $\bm{\pi}^{\text{ran}}$ is implemented into the RHS of~\eqref{Dtbound_multiuser} and $\bm{\pi}^{\text{DPP}}$ minimizes the RHS of~\eqref{Dtbound_multiuser}.
As $\{\theta_n\}_{n\in\mathcal{N}}$ are adjustable parameters, we let 
$
   \theta_n = \frac{\alpha_n(1 - ( 1 - \varepsilon_n) \eta_n )}{(1-\varepsilon_n) \eta_n}
$
for each $n \in \mathcal{N}$ and substitute them into the RHS of~\eqref{ref_multiuser} to obtain
\begin{align}  \label{phit3}
   {C}_t^{\bm{\pi}^{\text{DPP}}}
    &\leq \frac{V(\rho^2 + 1)}{2} + \sum_{n \in \mathcal{N}} \Big( \frac{\alpha_n}{(1-\varepsilon_n) \eta_n} + \alpha_n w_{n,t} \Big). 
\end{align}
Taking the expectation of both sides of~\eqref{phit3} with respect to the distribution of $\bm{o}_t$ yields 
\begin{align}  \label{phit4_multiuser}
   &\mathbb{E}^{\bm{\pi}^{\text{DPP}} }[ L_{t+1} - L_t ] + \sum_{n \in \mathcal{N}} \alpha_n \mathbb{E}^{ \bm{\pi}^{\text{DPP}} }[ h_{n,t+1}]  \notag \\
    &\leq \frac{V(\rho^2 + 1)}{2} + \sum_{n \in \mathcal{N}} \left( \frac{\alpha_n}{(1-\varepsilon_n) \eta_n} +  \alpha_n \mathbb{E}^{ \bm{\pi}^{\text{ran}} }\big[ w_{n,t} \big ] \right). 
\end{align}
Summing up~\eqref{phit4_multiuser} over $t \in \{0,1,2,\ldots,T-1\}$ and using the law of telescoping sums yields:
\begin{align}  \label{ELT}
    & \mathbb{E}^{ \bm{\pi}^{\text{DPP}} }[L_T] - \mathbb{E}^{ \bm{\pi}^{\text{DPP}} }[L_0 ] + \sum_{t = 1}^{T} \sum_{n \in \mathcal{N}} \alpha_n \mathbb{E}^{ \bm{\pi}^{\text{DPP}} }[ h_{n,t}] \notag \\
    & \leq \ B_1 T + B_2,
\end{align}
where $B_1 =  \frac{ V(\rho^2 + 1)}{2} + \sum_{n \in \mathcal{N}} \frac{\alpha_n }{(1-\varepsilon_n) \eta_n}$ and $B_2 =  \sum_{t = 0}^{T-1} \sum_{n \in \mathcal{N}} \alpha_n \mathbb{E}^{ \bm{\pi}^{\text{ran}} }\big[ w_{n,t} \big ] $.
By $\mathbb{E}^{\bm{\pi}^{\text{DPP}} }[h_{n,t} ] \geq 1$, we get
\begin{align}  \label{LT}
    \mathbb{E}^{ \bm{\pi}^{\text{DPP}} }[L_T ] 
    & \leq B_1 T + B_2 + \mathbb{E}^{ \bm{\pi}^{\text{DPP}} } [L_0 ].
\end{align}
By~\eqref{eq:lyapunov_multi}, we have 
\begin{align}  \label{EaLT}
    \mathbb{E}^{ \bm{\pi}^{\text{DPP}} } [L_T] = \frac{V\mathbb{E}^{ \bm{\pi}^{\text{DPP}} }[Q_T^2 ]}{2} + \sum_{n \in \mathcal{N}} \theta_n \mathbb{E}^{ \bm{\pi}^{\text{DPP}} } [ h_{n,T} ].
\end{align}
Substituting~\eqref{EaLT} into~\eqref{LT} yields
\begin{align}
    \frac{V\mathbb{E}^{ \bm{\pi}^{\text{DPP}} } [Q_T^2 ]}{2} 
    & \leq \frac{V\mathbb{E}^{ \bm{\pi}^{\text{DPP}} } [Q_T^2]}{2} + \sum_{n \in \mathcal{N}} \theta_n \mathbb{E}^{ \bm{\pi}^{\text{DPP}} } [ h_{n,T} ] \notag \\
    & \leq B_1 T + B_2 + \mathbb{E}^{ \bm{\pi}^{\text{DPP}} } [L_0 ].
\end{align}
Since $\left (\mathbb{E}^{ \bm{\pi}^{\text{DPP}} } [Q_T] \right)^2 \leq \mathbb{E}^{ \bm{\pi}^{\text{DPP}} } \left[Q_T^2 \right] $, 
we obtain 
\begin{align}   \label{limsupEQT}
    \limsup_{T \rightarrow \infty} \frac{\mathbb{E}^{ \bm{\pi}^{\text{DPP}} }\left[Q_T \right]}{T} 
    \leq \limsup_{T \rightarrow \infty} 
    \sqrt{\frac{2(B_1 T + B_2 + \mathbb{E}^{ \bm{\pi}^{\text{DPP}} }\left[L_0 \right])}{VT^2}}. 
\end{align}
Substituting $\mathbb{E}^{ \bm{\pi}^{\text{ran}} } \left[ L_0 \right] = \sum_{n \in \mathcal{N}} \theta_n$ obtained from~\eqref{eq:lyapunov_multi} and
\begin{align}  \label{eqwt}
\limsup_{T \rightarrow \infty}  \frac{B_2}{T^2} 
& = 
   \limsup_{T \rightarrow \infty}  \frac{1}{T^2} \sum_{t = 0}^{T-1} \sum_{n \in \mathcal{N}} \alpha_n \mathbb{E}^{ \bm{\pi}^{\text{ran}} }[ w_{n,t}] \notag \\
   & = \limsup_{T \rightarrow \infty}  \frac{1}{T}  \sum_{n \in \mathcal{N}} \alpha_n \frac{1-\lambda_n}{\lambda_n} = 0,
\end{align}
into~\eqref{limsupEQT}, we have $\limsup_{T \rightarrow \infty} \frac{\mathbb{E}^{ \bm{\pi}^{\text{DPP}} }\left[Q_T \right]}{T} = 0$.

\section{Proof of Theorem~\ref{thm:upp}} \label{proofoftheoremthmupp}
    Dividing both sides of~\eqref{ELT} by $T$ and letting $T$ approach $\infty$, and then substituting
\(
\limsup_{T \rightarrow \infty}  \frac{B_2}{T} 
=  \sum_{n \in \mathcal{N}} \alpha_n \frac{1-\lambda_n}{\lambda_n} ,
\)
    into the resulting equation yield
\begin{align}
    &\limsup_{T \rightarrow \infty}  \frac{1}{T} \sum_{t = 0}^{T-1} \sum_{n \in \mathcal{N}} \alpha_n \mathbb{E}^{ \bm{\pi}^{\text{DPP}} }[ h_{n,t}]
    \notag \\
    & \leq \limsup_{T \rightarrow \infty} \bigg( B_1 + \sum_{n \in \mathcal{N}}  \alpha_n  \frac{1-\lambda_n}{\lambda_n}  - \frac{\mathbb{E}^{ \bm{\pi}^{\text{DPP}}  } [L_T ]}{T} + \frac{\mathbb{E}^{ \bm{\pi} } [L_0]}{T} \notag \\
   &\quad\quad\quad\quad\quad\quad - \sum_{n \in \mathcal{N}} \alpha_n \Big(\frac{\mathbb{E}^{ \bm{\pi}^{\text{DPP}}  } [h_{n,T} ] } {T} 
    - \frac{\mathbb{E}^{ \bm{\pi}^{\text{DPP}}  } [h_{n,0} ] } {T} \Big) \bigg) \notag \\
    &\leq B_1 + \sum_{n \in \mathcal{N}}  \alpha_n  \frac{1-\lambda_n}{\lambda_n},  \label{lyapunov_upperbound}
\end{align}
where the last inequality is due to $\mathbb{E}^{ \bm{\pi}^{\text{DPP}}  } [ L_0 ] = \sum_{n \in \mathcal{N}} \theta_n $, $\mathbb{E}^{ \bm{\pi}^{\text{DPP}}  } [L_T ] \geq \sum_{n \in \mathcal{N}} \theta_n $ obtained from $\eqref{eq:lyapunov_multi}$, $\mathbb{E}^{ \bm{\pi}^{\text{DPP}}  } [h_{n,T} ] \geq 1$, and $\mathbb{E}^{ \bm{\pi}^{\text{DPP}}  } [h_{n,0} ] = 1$.
Then, recalling $B_1 =  \frac{ V(\rho^2 + 1)}{2} + \sum_{n \in \mathcal{N}} \frac{\alpha_n }{(1-\varepsilon_n) \eta_n}$, we minimize the upper bound as shown at the RHS of~\eqref{lyapunov_upperbound} through solving
\begin{equation} \label{problem_sub}
    \begin{aligned}
       \min_{\eta_n \in [0,\rho], \forall{n\in\mathcal{N}}} \quad &  {\sum_{n \in \mathcal{N}} \frac{\alpha_n}{(1-\varepsilon_n) \eta_n}}, \quad
   \text{s.t.} \  \sum_{n \in \mathcal{N}} \eta_n = \rho.
   \end{aligned}
\end{equation}
By the convex optimization theory~\cite{boyd2004convex},~\eqref{optimaletan} can be found as the solution of~\eqref{problem_sub}.

\bibliography{test}

\end{document}